\documentclass[11pt]{article}
\usepackage[final]{epsfig}
\usepackage[left=3cm,right=3cm, top=2.5cm,bottom=2.5cm,bindingoffset=0cm]{geometry}
\usepackage{graphics}
\usepackage{amsmath}
\usepackage{amsfonts}
\usepackage{latexsym}
\usepackage{amssymb, mathabx}
\usepackage{amsthm}
\usepackage{graphicx}
\usepackage{epstopdf}
\usepackage{multicol,multirow}
\usepackage{hyperref, enumitem}
\usepackage{mathtools}

\usepackage{mathrsfs}

 \usepackage{relsize}

\usepackage[nottoc]{tocbibind}

\usepackage{amsthm}

\usepackage[latin1]{inputenc}
\usepackage{tikz}
\usetikzlibrary{shapes,arrows,decorations.pathmorphing}
\usepackage{float}

\makeatletter
\def\thmhead@plain#1#2#3{%
  \thmname{#1}\thmnumber{\@ifnotempty{#1}{ }\@upn{#2}}%
  \thmnote{ {\the\thm@notefont#3}}}
\let\thmhead\thmhead@plain

\makeatother

\newcounter{AppCounter}

\def\restrict#1{\raise-.5ex\hbox{\ensuremath|}_{#1}}

\newtheorem{lemma}{Lemma}[section]
\newtheorem{proposition}[lemma]{Proposition}

\newtheorem{remark-definition}[lemma]{Remark-Definition}
\newtheorem{theorem}[lemma]{Theorem}
\newtheorem{corollary}[lemma]{Corollary}

\newtheorem{proposition-conjecture}[lemma]{Proposition-conjecture}

\theoremstyle{definition}
\newtheorem{example}[lemma]{Example}

\newtheorem{definition}[lemma]{Definition}
\newtheorem{remark}[lemma]{Remark}

\newcommand{\marginnote}[1]
{
}

\newcounter{cy}

\newcounter{bk}
\newcommand{\bk}[1]
{\stepcounter{bk}$^{\bf BK\thebk}$%
\footnotetext{\hspace{-3.7mm}$^{\blacksquare\!\blacksquare}$
{\bf BK\thebk:~}#1}}

\title{Averaging, symplectic reduction, and central extensions}

\author{Cheng Yang\thanks{Department of Mathematics and Statistics, McMaster University, Hamilton, ON L8S 4K1, Canada, and the Fields Institute for Research in Mathematical Sciences, Toronto, ON M5T 3J1, Canada; 
 e-mail: \tt{yangc74@math.mcmaster.ca}}~
  and Boris Khesin\thanks{Department of Mathematics, University of Toronto, Toronto, ON M5S 2E4, Canada; e-mail:  \tt{khesin@math.toronto.edu}
  }}

\date{Revised: October 2019}

\begin{document}

\maketitle
\begin{abstract}
We show that the averaged equation for a one-frequency fast-oscillating Hamiltonian system
is the result of symplectic reduction of a certain natural  system on the corresponding $S^1$-bundle
with respect to the circle action.
Furthermore, if the reduced configuration space happens to be a group, then under natural assumptions
the averaged system turns out to be the Euler equation on a central extension of that group. This gives a new explanation of the drift, common in averaged system, as a similar shift is typically present in symplectic reductions and central extensions.
\end{abstract}

\tableofcontents

\section{Introduction} \label{intro}

Dynamical systems with fast-oscillating conditions are ubiquitous in physics: they  naturally arise in mechanics,
astrophysics, fluid and air dynamics, and many other domains. They often exhibit surprising  properties, a
beautiful example of which is the inverted pendulum, which stabilizes via
fast vibration of its pivot. The standard way of analyzing such equations 
includes a procedure of constructing an averaged
system, whose solutions remain close to those of the original system for very long time 
(see e.g. \cite{kapi, land, ArnMM, arko}).

In many examples of Hamiltonian  one-frequency oscillating systems
one obtains an additional term, a drift in the averaged  equation. A similar drift (or shift) is observed
in  hydrodynamical-type systems, including the $\beta$-plane
equation in meteorology (see e.g. \cite{zeit}), infinite-conductivity equation for 
electron flows \cite{khch}, and the Craik-Leibovich equation for an ideal
fluid confined to a domain with oscillating boundary \cite{crle, vla}.
In those hydrodynamical systems
such a shift is often related to the consideration of a central extension of an appropriate Lie algebra \cite{viz}.

Below we explain this phenomenon by  building a general connection between the averaging method and symplectic 
reduction in appropriate, possibly nontrivial, $S^1$-bundles. Namely, one starts with
 symplectic reduction of the cotangent bundle over a circle action, which is one of the most studied objects
in symplectic geometry.
One observes two features for the reduction over a nonzero value of the momentum map: the appearance
of a twisted symplectic structure (similar to how the curvature 
arises in the description of gyroscopes on surfaces \cite{cole2}), 
where  a new magnetic term supplements the canonical symplectic form of the reduced space,
and the appearance of an amended potential function, see Section \ref{sec:reduction}. It turns out that exactly these two
 phenomena occur in the averaging procedure.  This can be summarized in the following statement (which is a combined version of Theorems \ref{thm:ave-red} and \ref{thm:Ave_natHamiltonian}):


\begin{figure}[H]
\tikzstyle{block} = [rectangle, draw, fill=blue!10,
    text width=6em, text centered, rounded corners, minimum height=6em]
\tikzstyle{line} = [draw, -latex']

\centerline{
\begin{tikzpicture}[node distance = 5.1cm,auto]
    \node [block] (SDE) {SDE (fast time=new space variable)};
    \node [block, right of=SDE] (tilSDE) {$\widetilde{SDE}$ (fibre-constant objects)};
    \node [block, below of=SDE] (aveDE) {$\overline{DE}_1$ average equation (over fast time)};
    \node [block, left of=aveDE] (DE) {DE with fast oscillation (fast time $\tau=\omega t$)};
    \node [block, below of=tilSDE] (redSDE) {$SDE_{red}$ (symplectic structure with magnetic term; effective potential)};
    \node [block, right of=redSDE, node distance=5cm] (Euler) {Euler equation on $\hat{\mathfrak{g}}^*$(central extension)};
    \path [line] (DE) -- node {suspension $\dot{\phi}=\omega$} (SDE) ;
    \path [line] (SDE) -- node [align=center]{space \\averaging \\for $S^1$-action} (tilSDE) ;
    \path [line,transform canvas={yshift=0.1em}]  (DE) -- node [align=center,below]{fast time\\averaging\\procedure}(aveDE);
    \path [ line, decorate, decoration=snake] (SDE) -- node [align=right,left] {Poincar\'e \\approximation\\theorem} (aveDE);
    \path [line] (tilSDE) -- node [align=right,left]{symplectic reduction  \\for a non-zero value \\of momentum map }(redSDE);
    \path [line,dashed] (tilSDE) -- node [align=right,xshift=-1.5em,yshift=-2em]{configuration space for \\$\widetilde{SDE}$ is the central \\extension group $\widehat G$, \\ reduction to $\hat{\mathfrak{g}}^*$ }(Euler);
    \path [line,dashed] (redSDE) -- node [align=center,below]{for group $G$\\ as the base,\\ Hamiltonian \\~ reduction to $\hat{\mathfrak{g}}^*$ }(Euler);
    \draw[implies-implies, double distance=0.5em](aveDE) -- (redSDE);
\end{tikzpicture}
}
\caption{DIAGRAM}	
\label{fig:diagram}
\end{figure}

\newpage 

\begin{theorem}
For a natural slow-fast Hamiltonian system  the resulting slow (averaged) system coincides with the one
obtained by space averaging over the fibers of an appropriate $S^1$-bundle and performing the symplectic reduction of the corresponding cotangent bundle over $S^1$-action at the momentum value related to the 
fast frequency. The averaged system turns out to be a natural Hamiltonian system with an amended potential function with respect to a twisted (magnetic) symplectic structure.
\end{theorem}

Furthermore, central extensions appear whenever the base of the reduction turns out to be a group by itself, as
discussed in Section \ref{sect:cenExtention}. This can be regarded as a manifestation of the reduction by stages
developed in \cite{mars2}.
The second main result of the paper is the following abbreviated version of Theorem \ref{thm:group-ext}:

\begin{theorem}
If the slow manifold is a group $G$ and the perturbed Hamiltonian system is invariant relative to the $G$-action,
then the second reduction of such 
a fast oscillating system gives an Euler equation, Hamiltonian with respect to  the
Poisson-Lie bracket on a central extension $\widehat{\mathfrak g}$ of the corresponding Lie algebra $\mathfrak g$.
\end{theorem}

The essence of the paper is described in the diagram in Figure \ref{fig:diagram}: we show how to view the fast time averaging approximation on the left by going via the averaging on the top and reduction in the right column of
the diagram. We describe this averaging-reduction procedure in Section \ref{sect:averaging}, and compare its result with the one obtained by using the classical fast-time averaging method in  Section \ref{sec:ODEaveraging}.

\smallskip

In Section \ref{sect:applications} we describe   three examples by using the averaging-reduction procedure developed in this paper: the vibrating pendulum manifests the appearance of the
amended potential, the Craik-Leibovich equation for oscillating boundary is related to the magnetic term in
the symplectic structure and a central extension, while  the motion of particles
in rapidly oscillating potentials has both  magnetic term and additional potential
present upon averaging.
(Note that, instead of the classical approach of applying ingenious canonical transformations  \cite{cole},
the present paper gives an alternative method of averaging natural Hamiltonian systems: one can average the metric, which contains all relevant information, and then obtain the averaged natural system directly
from that  metric.)

\smallskip

While the symplectic reduction  part of this paper is also valid for high-dimensional torus action, i.e. for
many-frequency case,   the approximation theorem does not work in this generality, as
for several frequencies resonances can appear unavoidably in such systems, 
as e.g. KAM theory manifests.
Note also that in many examples  the two contributions appearing in averaging, the magnetic term 
and the potential amendment, are of different order in the small parameter of perturbation. 
It would be interesting to see if it is always the case.

\bigskip

{\bf Acknowledgments.}
We are grateful to Mark Levi and Anatoly Neishtadt for stimulating discussions. A part of this work was done when C.Y. was visiting the Fields Institute in Toronto and B.K. was visiting the Weizmann Institute in Rehovot
and the IHES in Bures-sur-Yvette. The work was partially supported by an NSERC research grant.


\medskip

\section{Symplectic reduction of cotangent bundles}\label{sec:reduction}

We start by recalling  (following \cite{mars2}) general results on the symplectic
 reduction. Consider an action of  an abelian group $\mathbb T:=S^1$
(or more generally, a torus $\mathbb T=T^k$)
common in averaging, while the results   with appropriate amendments hold for a reduction by any Lie group.
Assume that the group $\mathbb T$ acts on a configuration space $Q$ (from the right) properly and freely,
so that the quotient space $Q/\mathbb T$ is a manifold.
Our first goal is to reduce $Q$ by the $\mathbb T$-action and describe structures on the reduced phase space.
The quotient
projection $\pi: Q \rightarrow B:=Q/\mathbb T$ defines a  principal fiber bundle over the base $B$.
It turns out that the curvature of this $\mathbb T$-bundle enters the symplectic structure of the reduced manifold.
The gyroscope example below can be regarded as an illustration of the abstract reduction procedure.
\smallskip

Namely, the group $\mathbb T$ acts on $T^* Q$ by cotangent lifts, and we denote the momentum map
 of this action by $J: T^* Q\rightarrow \mathfrak{t}^*$.
 The momentum map is a natural  projection 
 of  $T_q^* Q$  at any $q\in Q$ to the  cotangent space to the fiber, $\mathfrak{t}^*$.
For  an arbitrary value  $\mu\in \mathfrak{t}^*$ of the momentum map consider the reduced phase space\footnote{For an arbitrary Lie group the reduced space is defined as
$J^{-1}(\mu)/\mathbb T_\mu$ where $\mathbb T_\mu$ is the stationary subgroup of $\mu$.
In this section we use the fact that $\mathbb T$ is abelian, and hence  the stationary group $\mathbb T_\mu$ coincides with the full group: $\mathbb T_\mu = \mathbb T$.}
$(T^*Q)_\mu := J^{-1}(\mu)/\mathbb T$.

\begin{theorem} \label{thm:cotbundle}{\rm (see e.g. \cite{mars2})}
	Let $\mathbb T$ be an abelian group acting on a manifold $Q$ so that
	$\pi: Q \rightarrow Q/\mathbb T=:B$ is a principal fiber bundle, and  fix $\mu \in \mathfrak{t}^*$.
	Let $\mathcal{A}: TQ \rightarrow \mathfrak{t}$ be a principle connection 1-form on this bundle.
	Then
	
	 $i)$ for $\mu=0$ there is a symplectic diffeomorphism between $(T^* Q)_0$ and
	$T^*B=T^* (Q/\mathbb T)$ equipped with the canonical symplectic form 	$\omega_{can}$;
	
	$ii)$ for $\mu\not=0$ there is a symplectic diffeomorphism between $(T^* Q)_\mu$ and
	$T^*B$, where the latter is
	equipped with symplectic form $\omega_\mu:=\omega_{can} - \beta_\mu$. 	Here  the 2-form
	$\beta_\mu := \pi_P^* \sigma_\mu$ on $T^*B$ is obtained by the pull-back via the cotangent bundle
	projection $\pi_P: T^*B \rightarrow B$ from the 2-form $\sigma_\mu$ on $B$. The latter 2-form   is
	the $\mu$-component of the curvature of the principal fiber bundle $Q$ over $B$, namely
	$\pi^* \sigma_\mu = \mathbf{d}\left\langle \mu, \mathcal{A}\right\rangle.$
\end{theorem}

\begin{proof}[\textrm{\textit{Proof outline}}]
We just recall an explicit form of the isomorphism between $(T^* Q)_\mu$ and $T^* (Q/\mathbb T)$, see Theorem~2.3.3 in \cite{mars2} for more detail.

The isomorphism $\varphi_0: (T^* Q)_0 \rightarrow T^* (Q/\mathbb T)$ is defined by noting that
\[
J^{-1}(0) = \{  p_q \in T^* Q: \left< p_q , \xi_Q(q) \right> = 0 \quad \text{for all $\xi \in \mathfrak{t}$} \}\,,
\]
where $\xi_Q$ is the vector field on $Q$ corresponding to the infinitesimal action $\xi$, i.e. vectors $\xi_Q(q)$ span the vertical subspace at $q$.
Thus the map $\Phi : J^{-1}(0) \rightarrow T^*(Q/\mathbb T)$ given by
\begin{equation} \label{barphi}
	\left< \Phi(p_q) , \pi_*(v_q) \right> = \left<p_q, v_q \right>
\end{equation}
is well defined.  The map $\Phi$ is $\mathbb T$-invariant and surjective, and hence induces a quotient map $\varphi_0: (T^* Q)_0 \rightarrow T^* (Q/\mathbb T)$.

The isomorphism $\varphi_\mu : (T^* Q)_\mu \rightarrow T^* (Q/\mathbb T)$ is the composition
$\varphi_\mu = \varphi_0 \circ {\mathrm{shift}}_\mu$ of $\varphi_0$ with the 
 isomorphism  ${\mathrm{shift}}_\mu : (T^* Q)_\mu  \rightarrow (T^* Q)_0$  defined as follows.
 Introduce a map ${\mathrm{Shift}}_\mu : J^{-1}(\mu) \rightarrow J^{-1}(0)$ by
\[
	{\mathrm{Shift}}_\mu(p_q) = p_q - \left< \mu, \mathcal{A}(q) \right>
\]
for any $p_q\in  J^{-1}(\mu)$. It  is $\mathbb T$-invariant,  so  it drops to a quotient map ${\mathrm{shift}}_\mu: (T^* Q)_\mu \rightarrow (T^* Q)_0$.
The $\mathfrak{t}$-valued
2-form $\mathbf{d} \mathcal{A}$ is the curvature of the (abelian) connection $\mathcal{A}$, while to construct the 2-form  $\sigma_\mu$
one uses its $\mu$-component, cf.  \cite{mars2}.
\end{proof}

\begin{remark}
The isomorphism between $(T^* Q)_\mu$ and $T^*B=T^* (Q/\mathbb T)$ is connection-dependent.
The reduced symplectic form on $T^*B$ is modified by the curvature 2-form $\sigma_\mu$ on $B$, which is traditionally called a {\it magnetic term}, since it also appears in the description of  motion of a charged particle in a magnetic field on $B$.
\end{remark}

\begin{definition}\label{def:mech-conn}
Let the space $Q$ be  equipped with a $\mathbb T$-invariant metric. This metric
defines an invariant distribution of horizontal spaces: at each point $q\in Q$  there is a subspace of $T_qQ$ orthogonal to the fiber (i.e. the $\mathbb T$-orbit) at $q$. Hence the metric  defines an invariant  connection 1-form
$\mathcal{A}: TQ \rightarrow \mathfrak{t}$ on this fiber bundle. This 1-form is called a {\it mechanical connection.}
\end{definition}

Consider  a natural system on $T^*Q$ with Hamiltonian
$H(q,p)=(1/2) (p,\,p)_q+U(q)$
invariant with respect to the $\mathbb T$-action. (Here and below  $(.\,,.)_q$ stands for the
metric on $Q$, i.e. the inner product 
on $TQ$, or the induced one on  $T^*_qQ$, depending on the context. The Euclidean inner product in $\mathbb R^n$ is denoted by dot.) 
This system descends to a Hamiltonian system
on the quotient $(T^* Q)_\mu$ with respect to the symplectic structure  $\omega_\mu=\omega_{can} - \beta_\mu$.
The new Hamiltonian $H_\mu$ is obtained from $H$ by applying the map
${\mathrm{Shift}}_\mu$ and the corresponding potential $U(q)$ acquires an additional term, as we discuss below.

\begin{example}\label{ex:spinDisk}
The following example of a spinning disk (a gyroscope) on a curved surface shed the light on the geometry behind the symplectic reduction above. 
 Cox and Levi proved in \cite{cole2}  that the motion of the disk center coincides with the motion of a charged particle in a magnetic field which is normal to the surface and  equal in magnitude to the Gaussian curvature of the surface. 
 
To explain their result in the context of  reduction theory,  let $q=(q_1,q_2)$ be orthogonal local coordinates on a surface 
$B\subset \mathbb R^3$, so that  metric on the surface is given by $ds^2=a_{11}(q)dq_1^2+a_{22}(q)dq_2^2$. 
When the disk is not spinning, its kinetic energy is a function of its position $q$ and linear velocity $\dot q$, i.e. a function on $TB$. It  is given by
$$
E_0=\frac m2(G\dot q, \dot q)+\frac{\mathbb I_d}2 h(\dot q,\dot q),
$$
where $G=diag(a_{11},a_{22})$, $h$ is the second fundamental form of $B$, and $\mathbb I_d$ is the moment of inertia of the disk along its diameter. Denote by $\mathbb I_a$ the moment of inertia about the disk axis.

\begin{theorem}\label{thm:spinDisk}{\rm \cite{cole2}}
For a spinning disk on a surface $B$ the angular momentum $\mu=\mathbb I_a\omega_a$ of the disk about its axis is  constant and  the disk's center satisfies the following equation 
\begin{equation}\label{eq:spinDisk} 
\frac {d}{dt}\frac{\partial E_0}{\partial \dot q}-\frac{\partial E_0}{\partial q}=\sqrt{a_{11}a_{22}}\,\mu\,K(q)\,\left[\begin{array}{cc}0&-1\\1&0\end{array}\right]\,\dot q,
\end{equation}
where $K(q)$ is the Gaussian curvature of the surface $B$. 
\end{theorem}

\begin{remark}
Before we provide a different proof of this result via symplectic reduction, note that the configuration space $Q$ of this system is the circle bundle over the surface: the disk position is defined by the position of its center on the surface $B$ 
and the angle of rotation.
Globally $Q$ may be a nontrivial $\mathbb T$-bundle over $B=Q/\mathbb T$. However,  for the local consideration below it suffices to consider the case 
$Q=B\times \mathbb T$.  The phase space  is the corresponding tangent bundle $TQ$. 
The metric in $\mathbb R^3$ allows one to identify $TQ$ and $T^*Q$, while the disk motion is a Hamiltonian system on the cotangent bundle $T^*Q$.
The trajectory of the disk center can be obtained as the symplectic reduction  of the system on $T^*Q$ with respect to the
$\mathbb T$-action, as we quotient out the disk rotation. Different angular velocities of the disk lead to different values $\mu$
of the momentum map, over which  one takes the quotient. 
According to Theorem \ref{thm:cotbundle} the resulting system is a Hamiltonian system on $T^*B$, with two 
amendments. The corresponding symplectic structure after the reduction will be twisted by a magnetic term.
In this setting it will be proportional to the curvature $K(q)$ of the surface, which one observes in Equation \eqref{eq:spinDisk}. 
Moreover, the corresponding Hamiltonian undergoes a shift by $\mu$. However, in the gyroscope case the shift 
reduces to adding a constant to the Hamiltonian and does not appear in the equations.\footnote{Jumping ahead, in order to see this 
one can use an explicit formula of Theorem \ref{thm:ave-red}, which gives an additional term $\frac 12 \langle \mu,\,\mathbb{I}(q)^{-1}\mu\rangle$. It is indeed constant, since in the gyroscope case the inertia operator $\mathbb I$ does not depend on $q$, while $\mu$ is a constant angular velocity.}
\end{remark}

\begin{proof}
Let us identify $TB$ and $T^*B$ by means of the metric. 
First note that Equation (\ref{eq:spinDisk}) is the Euler-Lagrange equation for a Lagrangian system, which can be rewritten as a Hamiltonian system with the  Hamiltonian energy function $E_0$ on the cotangent bundle of the surface $B$ (thanks to the metric identification) with a twist symplectic structure given in local coordinates by
$$
\omega_{\mu}=\omega_{can}-\mu\,\sqrt{a_{11}a_{22}}\,K(q)\,dq_1\wedge dq_2\,,
$$
where $\omega_{can}$ is the canonical symplectic structure on $T^*B$.

Next, we show how to obtain this system via symplectic reduction.
Denote by $\theta$  the angle between a fixed radius on the disk and the positive direction of the line 
$\{q_2=const\}$. This gives us a principal
$\mathbb T$-bundle $Q$ with the curved surface $B$ as the base. 

The absolute angular velocity of a spinning disk is $\Omega_a=\dot \theta+A(q)\dot q$, where $A(q)\dot q$ is the transferred velocity, and $A(q)=(k_1\sqrt{a_{11}},k_2\sqrt{a_{22}})$, where $k_1,\;k_2$ are the geodesic curvatures of coordinate lines
$\{q_1=const\}$ and $\{q_2=const\}$.

So, in local coordinates, the metric on the principal $\mathbb T$-bundle $Q$ is given by
$$
\left((\dot q,\dot \theta),(\dot q,\dot \theta)\right)_{(q,\theta)}=\mathbb I_a(\dot\theta+A(q)\dot q)^2+m(G\dot q,\dot q)+\mathbb I_dh(\dot q,\dot q).
$$
Note that this metric is invariant under the $\mathbb T$-rotations.

Therefore, the momentum map $J:TQ\rightarrow \mathfrak t^*=\mathbb R$ is $J(q,\theta;\dot q,\dot \theta)=\mathbb I_a(\dot\theta+A(q)\dot q)$ and the mechanical connection
$\mathcal A:TQ\rightarrow\mathfrak t=\mathbb R$ is  $\mathcal A=d\theta+A(q)\,dq$ (here we again identify $TQ$ and $T^*Q$). By Theorem \ref{thm:cotbundle}, for a fixed value $\mu=\mathbb I_a\Omega_a$ of
the momentum map $J$, the system can be reduced to the (co)tangent bundle of the surface $B$ with the magnetic symplectic structure
$$
\omega_{\mu}=\omega_{can}-\mu\;d(A(q)dq)=\omega_{can}-\mu\;\sqrt{a_{11}a_{22}}\,K(q)\,dq_1\wedge dq_2,
$$
where $K(q)$ is the Gaussian curvature of the surface $B$. 

The energy Lagrangian $E=\frac 12 \left((\dot q,\dot \theta),(\dot q,\dot \theta)\right)_{(q,\theta)}$ 
on $Q$ defines the reduced Hamiltonian on $B$, which  turns out to be $E_0=1/2(m(G\dot q, \dot q)+\mathbb I_d h(\dot q,\dot q))$.  Here we omit the constant term $\mathbb I_a(\dot\theta+A(q)\dot q)^2
=\langle \Omega_a,\,\mathbb{I}_a\Omega_a\rangle= \langle \mu,\,\mathbb{I}_a^{-1}\mu\rangle$ in the energy expression, 
since the value $\mu$ of the momentum map (i.e. the angular momentum of the disk) is conserved. 

This reduced Hamiltonian system with Hamiltonian function $E_0$ 
on the cotangent bundle $(T^*B, \omega_\mu)$ of the surface describes the motion of the disk center.
\end{proof}

\end{example}


\medskip
\section{Averaging-Reduction procedure for a natural system}\label{sect:averaging}
\subsection{Averaging}
Let $\pi:Q\to B$ be a principal $\mathbb T$-bundle. From now on we assume that $\mathbb T=S^1$ (and occasionally comment on $\mathbb T=T^n$). The cotangent lift of $\mathbb T$-action on $Q$ induces $\mathbb T$-action on $T^*Q$.
Denote by $\rho, \rho^*$, and $\rho_*$ the $\mathbb T$-action on $Q, T^*Q$ and $TQ$, respectively.
Let $d\eta$ be the standard Euclidean measure on the group $\mathbb T$.

Consider a natural Hamiltonian system on the cotangent bundle $T^*Q$:
\begin{equation}\label{eq:nature}
H(q,p)=\frac 12 (p,\,p)_q+U(q).
\end{equation}
Here $Q$ is the configuration space of the motion, we assume that this Hamiltonian system has slow motion on the base  $B$ and fast motion on the fibers isomorphic to $\mathbb T$. 
The Hamiltonian function $H(q,p)$ is not necessarily invariant
under the $\mathbb T$-action on $T^*Q$. As the first step one passes to the space $\mathbb T$-average
$\overline{H(q,p)}^{\mathbb T}$, the $\mathbb T$-invariant function on $T^*Q$ defined by the following formula:
$$
 \overline{H(q,p)}^{\mathbb T}:=\frac{1}{\eta(\mathbb T)}\int_{g\in \mathbb T} H(\rho_g^*(q,p))\;d\eta(g).
$$
For the natural system \eqref{eq:nature}, one averages both the kinetic and potential parts of the energy:
$$
 \overline{H(q,p)}^{\mathbb T}=\frac 12 \overline{(p,p)}_q^{\mathbb T}+\overline{U(q)}^{\mathbb T},
$$
where  $\overline{U(q)}^{\mathbb T} =\frac{1}{\eta({\mathbb T})}\int_{\mathbb T} U(\rho_g(q))\,d\eta(g)$
and $\overline{(p,p)}_q^\mathbb T=\frac{1}{\eta({\mathbb T})}\int_{\mathbb T}(\rho_g^*p,\rho_g^*p)_{\rho_{g^{-1}}(q)}\,d\eta(g)$ 
is defined via the following averaged metric on $Q$:
\begin{definition}
 The \textit{averaged metric} $\overline{(\cdot,\cdot)}^{\mathbb T}$ on the principal $\mathbb T$-bundle $Q$ is given by
$$
\overline{(v,v)}_q^{\mathbb T}:=\frac{1}{\eta(\mathbb T)}\int_{\mathbb T}(\rho_{g*} v,\rho_{g*} v)_{\rho_g(q)}\,d\eta(g),
$$
for any $v\in T_qQ$. This defines a  $\mathbb T$-invariant metric on $Q$.
\end{definition}

Now define the connection on $Q$ corresponding to the averaged metric:

\begin{definition}
 The \textit{averaged connection} $\bar{\mathcal{A}}\in\Omega^1(Q,\mathfrak t)$ on the principal $\mathbb T$-bundle $Q$ is the  connection  induced by the averaged metric $ \overline{(\,,\,)}_q^{\mathbb T}$
by the invariant distribution of horizontal spaces: at each point $q\in Q$  there is a subspace of $T_qQ$ orthogonal to the fiber (i.e. the $\mathbb T$-orbit) at $q$.  

The connection induced by an invariant averaged metric on $Q$ is  the mechanical connection, according to
Definition \ref{def:mech-conn}.
\end{definition}

\begin{remark}
We would like to give a more explicit description of averaged metrics and connections.
First note that a  $\mathbb T$-invariant metric $(.\,,.)_q$ on $Q$ can be defined by means of a metric operator
$\mathbb I_Q(q):T_qQ\to T^*_qQ$ for $q\in Q$, where $\mathbb I_Q(q): v\mapsto v^\flat$, i.e. 
$(v,v)_q:=\langle v, \mathbb I_Q (q)v\rangle$ for $v\in T_qQ$.
This defines  the ``fiber inertia operator" $\mathbb I(q):\mathfrak t\to \mathfrak t^*$ by restricting to
$\mathfrak t=T_q\mathbb T\subset T_qQ$ the metric operator $\mathbb I_Q(q)$
for $q\in Q$. (Recall, that for $\mathbb T=S^1$, we have $\mathfrak t=\mathbb R$.) The $\mathbb T$-invariance of metric implies that the fiber inertia operator $\mathbb I$ is equivariant, $\mathbb I(g(q))=Ad_{g^{-1}}^*\mathbb I(q)$, i.e. it depends on the base point $\pi(q)\in B=Q/\mathbb T$ only.

\smallskip
The invariant metric on $TQ$ also induces the momentum map $J:T^*Q\to \mathfrak t^*$ for the action of the group $\mathbb T$.
In these terms the averaged mechanical connection can be defined explicitly by
$$
\bar{\mathcal{A}}(v_q)=\mathbb{I}(q)^{-1}J(p_q),
$$
where $v_q$ is a tangent vector in $T_qQ$, \, $p_q:=\mathbb I_Q(q)v_q=v_q^\flat\in T_q^*Q$ is the corresponding metric-dual cotangent vector, and $\mathbb I$ is the inertia operator on $\mathfrak t$ in the fiber at $q$.
\end{remark}

\begin{remark}\label{rem:metric}
More specifically, in coordinates for a trivial bundle $Q$  the general form for a $\mathbb T$-invariant metric on $Q=B\times \mathbb T$ is as follows:
\begin{equation}\label{eq:aveMet}
((u,\gamma),(u,\gamma))_{(x,\tau)}=(u,u)_x+2\gamma\, h(x)\,\langle A(x),u\rangle+h(x)\gamma^2,
\end{equation}
where $(u,\gamma)\in T_{(x,\tau)}(B\times \mathbb T)=T_x B\times\mathfrak t$, $A(x)\in \Omega^1(B,\mathfrak t)=T_x^*B$,   $h(x)\in \mathbb R^{+}$, and $\mathfrak t\simeq \mathbb R$.
For a non-trivial $Q$ this general form is valid locally on the base.
\end{remark}

\begin{proposition}\label{prop:momentum}
For a trivial bundle $Q=B\times\mathbb T$ the averaged connection 
$
\bar{\mathcal{A}}\in\Omega^1(B\times\mathbb T,\mathfrak t)= T^*_{(x,\tau)}(B\times\mathbb T)
$ 
corresponding to the averaged metric (\ref{eq:aveMet}) is given by $\bar{\mathcal{A}}(x,\tau)=A(x)+d\tau$.
The summands can be regarded as connections on the base $A(x)\in T^*_xB$   and in the fiber $d\tau$.
\end{proposition}

\begin{proof}
For a trivial bundle $Q$  the momentum map $J:T_{(x,\tau)}^*(B\times \mathbb T)\to \mathfrak t^*$ is given by 
$$
J_{(x,\tau)}(a,\eta)=h(x)\langle A(x),u\rangle+h(x)\langle d\tau, \xi\rangle,
$$
where $(a,\eta)\in T_{(x,\tau)}^*(B\times \mathbb T)$ and $(u,\xi)\in T_{(x,\tau)}(B\times \mathbb T)$ is the image 
of $(a,\eta)$ under the metric identification. Indeed, by  definition of the momentum map, 
for any $\zeta\in\mathfrak t$, 
$$
\langle J_{(x,\tau)}(a,\eta),\zeta\rangle=\langle (a,\eta),(0,\zeta)\rangle=((u,\xi),(0,\zeta))
$$
$$
=(u,0)_x+\zeta h(x)\langle A(x),u\rangle+\xi h(x)\langle A(x),0\rangle+\xi h(x)\zeta=\langle h(x)\langle A(x),u\rangle+h(x)\xi,\zeta\rangle.
$$  
Furthermore, the inertia operator $\mathbb I(x):\mathfrak t\to \mathfrak t^*$ at $x\in B$ is given by
$\mathbb I(x)\gamma=h(x)\gamma $ for any $\gamma\in\mathfrak t,$ hence 
the average mechanical connection assumes the form
$\bar{\mathcal{A}}(u,\xi)=\mathbb{I}(x)^{-1}J(a, \eta)=\langle A(x),u\rangle+\langle d\tau, \xi\rangle$, as required.
\end{proof}

\smallskip


\subsection{Reduction}
By considering the $\mathbb T$-invariant metric and Hamiltonian (obtained by $\mathbb T$-averaging)
we are now in the framework of Section \ref{sec:reduction}.
The dynamics defined by the averaged Hamiltonian $\overline{H}^{\mathbb T}$
on $T^*Q$ can be derived from the corresponding \textit{averaged} or \textit{slow motion}, i.e. 
the dynamics on $T^*B$ of the base space $B=Q/\mathbb T$. However, unlike the standard averaging  
method discussed below in Section \ref{subsec:averHam}, now we obtain this slow motion via symplectic reduction.

Recall that, for a fixed value $\mu$ of the momentum map,  the reduced space $J^{-1}(\mu)/\mathbb T$ is symplectomorphic to the cotangent bundle $T^*B$ of the base space $B$ with the twisted symplectic form 
 $\omega_\mu=\omega_{can}-\beta_\mu,$
where $\omega_{can}$ and $\beta_\mu$ are the canonical and magnetic 2-forms
on $T^*B$ (see Theorem \ref{thm:cotbundle}).
The averaged/slow system turns out to be a Hamiltonian system on the symplectic manifold $(T^*B,\omega_\mu)$ with the following reduced Hamiltonian function $\bar H_\mu$.

\begin{theorem}\label{thm:ave-red}
For a natural system on a $\mathbb T$-bundle $Q$ over the slow manifold $B$
with Hamiltonian function $H(q,p)=(1/2) (p,\,p)_q+U(q)$
the result of the symplectic reduction
with respect to the $\mathbb T$-action of the averaged system is a natural system
with the Hamiltonian function $\bar H_\mu$,
\begin{equation}\label{aveHam}
\bar H_\mu(q,p)=\frac 12 (p,\,p)_B+U_\mu(q)\,,
\end{equation}
 on the symplectic manifold $(T^*B,\omega_\mu)$.
Here $(q,p)\in T^*B$, $(\cdot,\cdot)_B$ stands for the metric on the base $B=Q/\mathbb T$ obtained as a Riemannian
submersion $Q\to B$ from the metric 
$\overline{(\cdot,\cdot)}^{\mathbb T}$ on $Q$, while 
$U_\mu(q):=\overline{ U(q)}^{\mathbb T}+\frac 12 \langle \mu,\,\mathbb{I}(q)^{-1}\mu\rangle$ is the effective potential.
\end{theorem}

\begin{proof}
We start by computing the result of averaging and consequent symplectic reduction  on $T^*Q$ with respect to the $\mathbb T$-action.
Upon averaging along $\mathbb T$-orbits one can assume that the Hamiltonian $\bar H$ on $Q$ is $\mathbb T$-invariant,
$\bar{H}(q,p)=\overline{ H(q,p)}^{\mathbb T}$.
The reduced Hamiltonian system on the quotient $(T^* Q)_\mu$ is Hamiltonian with respect to the symplectic structure  
$\omega_\mu=\omega_{can} - \beta_\mu$. The new Hamiltonian is obtained from $\bar H$ by applying the map
${\mathrm{Shift}}_\mu$.
Namely, abusing the notation, for $(q,p)\in T^*B$  and a connection  $\bar{\mathcal{A}}$ in the $\mathbb T$-bundle $Q$ one has
$$
\bar H_\mu(q,p)=\bar H (q,\,p+  \langle \mu,\bar{\mathcal{A}}(q)  \rangle)
=\frac 12\overline{( p+  \langle \mu, \bar{\mathcal{A}}(q)  \rangle, \, p+  \langle \mu,\bar{\mathcal{A}}(q)\rangle )}_q^{\mathbb T}+\overline{U(q)}^{\mathbb T}
$$
$$
=\frac 12 (p,\,p)_B+\overline{(p,\,  \langle \mu, \bar{\mathcal{A}}(q)  \rangle)}_q^{\mathbb T}
+ \frac 12\overline{(\langle \mu, \bar{\mathcal{A}}(q)\rangle,\,\langle \mu, \bar{\mathcal{A}}(q)\rangle)}_q^{\mathbb T}+\overline{U(q)}^{\mathbb T}
=\frac 12 (p,\,p)_B+\overline{U_\mu(q)}^{\mathbb T}
$$
for $U_\mu(q):=\frac 12 \overline{(\langle \mu, \bar{\mathcal{A}}(q)\rangle,\,\langle \mu, \bar{\mathcal{A}}(q)\rangle)}_q^{\mathbb T}+\overline{U(q)}^{\mathbb T}$.
Here we use that $\bar{\mathcal{A}}$ is the mechanical connection corresponding to the averaged metric $\overline{(\cdot,\cdot)_q}^{\mathbb T}$, and hence we have
$\overline{(p,\,  \langle \mu, \bar{\mathcal{A}}(q)  \rangle)}_q^{\mathbb T}=\langle \mu, \, \bar{\mathcal{A}}(q)(v)\rangle=\langle \mu,\,\mathbb I(q)^{-1}J(p)\rangle=0$, since $J(p)=0$, and where $(q,p)\in T^*B$ is identified with  $(q,v)\in TB$ by means of
the averaged metric. Thus on the reduced symplectic manifold
$T^*B$ with the twisted symplectic form $\omega_\mu=\omega_{can} - \beta_\mu$
the new reduced Hamiltonian is
$$
\bar H_\mu(q,p)=\frac 12 (p,\,p)_B+U_\mu(q)
$$
for $q\in B$ and $p\in T^*_qB$. It is  a natural system with a new effective potential
$$
U_\mu(q)=\frac 12 \overline{(\langle \mu, \bar{\mathcal{A}}(q)\rangle,\,\langle \mu, \bar{\mathcal{A}}(q)\rangle)}_q^{\mathbb T}+\overline{U(q)}^{\mathbb T}
=\frac 12 \langle \mu, \mathbb I(q)^{-1}\mu\rangle+\overline{U(q)}^{\mathbb T}\,.
$$
\end{proof}

In Section \ref{sect:Ave_natHamiltonian} below we will prove the following corollary of Theorem \ref{thm:ave-red} for averaging
one-frequency fast-oscillating systems: Under certain conditions, solutions of the averaged system and projections to slow manifold of solutions of the actual system with the same initial conditions remain $\epsilon$-close to each other for $0\le t\le 1/\epsilon$.

\begin{remark}
The two features of the averaged-reduced Hamiltonian system are the additional term in the effective potential
$U_\mu$ and the magnetic term $-\beta_\mu$ in the symplectic structure $\omega_\mu$.
Therefore this averaging-reduction procedure provides a geometrical explanation of these two phenomena,
often observed in the averaging theory.
\end{remark}

\begin{remark}
In the classical averaging of fast-oscillating systems  (cf. Section \ref{subsec:averHam} below)
one starts by fixing the action variable $J$. This  
can be regarded as a manifestation of symplectic reduction in flat coordinates, as this means fixing 
a certain value of the corresponding momentum map.
The bundle averaging-reduction procedure described here is also applicable in that case, but the metric
in this bundle turns out to be flat. Namely, in the reduction to a submanifold $J=\mu$ 
one chooses a flat connection on the principal bundle which corresponds to the direct product of the base
and fibres, and hence no twisted symplectic structure appears on the reduced manifold: for
the momentum value $J=\mu$, the averaged Hamiltonian function is $\epsilon \,\bar{H}(Q,P,\mu)$
on the ``flat" cotangent bundle  $(T^*\mathbb R^{\ell}, dP\wedge dQ)$.
\end{remark}


\medskip

\section{Central extensions in symplectic reduction} \label{sect:cenExtention}
Above we described the reduced phase space $(T^* Q)_\mu$ for the right  action by the group $\mathbb T$.
In this case, the reduced phase space $(T^* Q)_\mu$ coincides with $T^* (Q/\mathbb T)$, equipped with the magnetic symplectic structure $\omega_\mu$ described before.
Now assume in addition that the base space $Q/\mathbb T$ has the structure of another Lie group $G$, which
acts on itself from the left and leaves the metric on $G=Q/\mathbb T$ invariant.
 As a result, $G$ acts on $T^* G=T^* (Q/\mathbb T)$ and, as one can check,  this action leaves the symplectic structure $\omega_\mu=\omega_{can} - \beta_\mu$ invariant.
 (Recall that the magnetic 2-form $\beta_\mu := \pi_G^* \sigma_\mu$
 on $T^* G$ is the pullback of the left-invariant 2-form $\sigma_\mu$ on the group $G$.)
 Hence another  reduction  for this $G$-action (``the reduction by stages") would take this magnetic
 symplectic structure on $T^*G$
 to an appropriate structure on the dual Lie algebra $\mathfrak{g}^*$, as described below.

\begin{theorem}{\rm (Theorem~7.2.1 in \cite{mars2})}\label{thm:poisson}
  The Poisson reduced space for the left  action  of
  $G$ on $(T^* G, \omega_\mu=\omega_{can} - \beta_\mu)$ is the dual Lie algebra $\mathfrak{g}^*$ with the
  Poisson bracket given by
\begin{equation} \label{bracket}
\{f, g\}_\mu(\nu) = -\left< \nu, \left[ \frac{\delta f}{\delta \nu},
    \frac{\delta g}{\delta \nu} \right] \right> - {\sigma}_\mu(e)\left(
  \frac{\delta f}{\delta \nu}, \frac{\delta g}{\delta \nu} \right)
\end{equation}
for $f, g \in C^\infty(\mathfrak{g}^*)$ at any $\nu\in \mathfrak{g}^*$, where ${\sigma}_\mu(e)$
is the value of the  left-invariant  2-form ${\sigma}_\mu$ at $e\in G$ on the pair of tangent vectors  $\frac{\delta f}{\delta \nu}, \frac{\delta g}{\delta \nu}\in T_eG=\mathfrak{g}$, and $\beta_\mu := \pi_G^* \sigma_\mu$ is the pullback
of ${\sigma}_\mu$ to $T^*G$.
\end{theorem}

\begin{remark}
The above Poisson bracket is the Lie-Poisson bracket of the dual $\widehat{\mathfrak{g}}^*$ of the central extension
$\widehat{\mathfrak{g}}$ of the Lie algebra ${\mathfrak{g}}$ by means of the ${\mathfrak{t}}$-valued
2-cocycle $\sigma$, such that $\langle \sigma, \mu\rangle :={\sigma}_\mu(e)$.
Namely, the Lie algebra $\widehat{\mathfrak{g}}$ is the direct sum ${\mathfrak{g}}\oplus {\mathfrak{t}}$,
as a vector space, with  the commutator
$$
[(u,a), (v,b)]_{\widehat{\mathfrak{g}}}:= ([u,v]_{\mathfrak{g}}, {\sigma}(u,v))
$$
for  $u,v\in \mathfrak{g}$ and $a,b\in {\mathfrak{t}}$.
It turns out that under certain integrality conditions, the space $Q$ gives a realization of the corresponding
centrally extended group $\widehat G$.
\end{remark}

For the right action of $G$ the bracket changes sign.

\begin{theorem}\label{thm:group-ext}
Let $G$ be a group equipped with a closed integral left-invariant 2-form $\sigma_\mu/2\pi$. Then the $\mathbb T$-bundle $Q$
over the group $G$ with the curvature form $\sigma_\mu$ can be canonically identified with the
central extension $\widehat G$
of the group $G$ by means of $\mathbb T$, where the Lie algebra 2-cocycle is ${\sigma}_\mu(e)$, i.e.
its value on a pair of Lie algebra elements $\xi$ and $\eta$ is ${\sigma}_\mu(e)(\xi,\eta)$.
\end{theorem}

\begin{proof} The proof is based on a version of Proposition 4.4.2 of \cite{segal} adjusted to the setting at hand.
In fact, one can explicitly construct $\widehat G$ and identify it with $Q$, the $\mathbb T$-bundle over $G$.
Namely, first for any oriented loop $\ell$ in $G$  one associates an element
$C(\ell)=\exp (i\int_{\partial^{-1}\ell}\sigma_\mu)$, where $\partial^{-1}\ell$ is an oriented 2D surface in $G$ bounded
by $\ell$. The value $C(\ell)$ is well-defined, since for two different surfaces with the same boundary
the integrals of $\sigma_\mu$ for an integral 2-form $\sigma_\mu/2\pi$ differ by a multiple of $2\pi$.

The map $\ell\mapsto C(\ell)$ is independent of parametrization of $\ell$, additive, and $G$-invariant.
It defines a central extension $\widehat G$ of the group $G$ by $\mathbb T$ as a set of triples $(g,u, p)$, where
$g\in G$, $u\in \mathbb T$ and $p$ is a path in $G$ from $e$ to $g$, modulo the following equivalence.
Two triples  $(g,u, p)$ and  $(g',u', p')$ are equivalent if $g'=g$ and $u'=C(p'\cup p^{-1})u$. The composition
is $(g_1,u_1, p_1)\circ (g_2,u_2, p_2)= (g_1g_2, u_1u_2, p_1\cup g_1(p_2)))$.

Recall that the space $Q$ with an invariant metric has a structure of 
a $\mathbb T$-bundle with mechanical connection. Then a triple
$(g,u, p)$ modulo equivalence can be interpreted as the following point in $Q$: it is the point in the
$\mathbb T$-fiber over  $g\in G$, obtained from the point $(e,u)$ of the $\mathbb T$-fiber over $e\in G$
by a horizontally lifted path $p$ from $e$ to $g$. Then the equivalence of triples stands for their correspondence to the same point in $Q$, since the form $\sigma_\mu$ is the curvature of the mechanical connection, while
the formula $u'=C(p'\cup p^{-1})u$ describes the holonomy of the connection over a closed loop.
\end{proof}

\begin{remark}
Theorems \ref{thm:poisson} and \ref{thm:group-ext} can be extended to  
the case of a torus $\mathbb T$-bundle $Q$ over $G$, where $\mathbb T=T^k$. In Theorem  \ref{thm:group-ext}
one realizes $Q$ as a group  central extension of $G$ by $\mathbb T$ by applying the above consideration to 
the ``coordinate 2-forms"  $\sigma_\mu=\langle \sigma, \mu\rangle $ of the $\mathfrak t$-valued 2-form $\sigma$.
\end{remark}

\begin{remark}
Return to the 2-cocycle $\beta_\mu$ on the Lie algebra $\mathfrak g$, which defines
the central extension and the magnetic term. In many examples, this 2-cocycle
is a 2-boundary, i.e. the 2-form $\sigma_\mu$ on the Lie algebra can be represented as a linear
functional  of the Lie algebra commutator,  $\sigma_\mu(\xi, \eta)=L([\xi,\eta])$
for some element $L\in \mathfrak g^*$.
In that case, the corresponding Poisson structure on $\mathfrak g^*$ is the linear Lie-Poisson structure on the dual space $\mathfrak g^*$ shifted to the point $L$.
The associated Euler equation also manifests a certain shift, observed, e.g.
as a Stokes drift velocity related to surface waves in the  Craik-Leibovich equation, cf. Section \ref{sect:CL}.
\end{remark}

\begin{remark}
When considering dynamics on the reduced space $T^*G$, in order to use the second reduction
over the $G$-action one has to confine to the natural systems with effective potential independent of $q$, i.e.
$U_\mu(q)=const$. The latter are geodesic flows for the invariant metric on $G$ defined by the inertia operator
$\mathbb I_G:\mathfrak g\to \mathfrak g^*$. The second reduction 
 defines the Euler equations for the quadratic Hamiltonian
$H(p):=\frac 12 (p,p)_e=\frac 12 \langle p, \mathbb I_G^{-1}p\rangle$
on the dual $\widehat{\mathfrak{g}}^*$ of  centrally extended Lie algebra $\widehat{\mathfrak{g}}$.
\end{remark}

%

\medskip

\section{Reminder on averaging and examples}\label{sec:ODEaveraging}

\subsection{Averaging in one-frequency Hamiltonian systems}\label{subsec:averHam}
Consider a Hamiltonian system with $\ell+1$ degrees of freedom and Hamiltonian of the form
$H(q,p, I, \phi)=H_0(I)+\epsilon H_1(q,p, I, \phi)$, where $\phi ({\rm mod }\, 2\pi)\in \mathbb T$, 
while $H$ is $2\pi$-periodic in $\phi$, and $(q,p, I)\in D\subset \mathbb R^{2\ell+1}$. (Such perturbations of properly degenerate Hamiltonian  systems are typical in celestial mechanics.)
The corresponding Hamiltonian equations for the standard symplectic structure are as  follows:
\begin{equation}\label{eq:Ham_eq}
\left\{
  \begin{array}{l}
        ~ \dot q = \quad \epsilon\,{\partial H_1}/{\partial p} \qquad\quad
		\dot I = \qquad\qquad  -\epsilon \,{\partial H_1}/{\partial \phi},\\\
		              \dot p = -\epsilon\,{\partial H_1}/{\partial q}, 
		              \qquad\quad\dot \phi=\,{\partial H_0}/{\partial I}+\epsilon\,{\partial H_1}/{\partial I}.
\end{array} \right.
\end{equation}

\begin{definition}
The {\it averaged system} for the above Hamiltonian $H=H_0+\epsilon H_1$ is the system of $2\ell+1$ equations:
\begin{equation}\label{eq:Ave_Ham_eq}
\left\{
  \begin{array}{l}
				~\dot Q = \quad \epsilon \,{\partial \bar{H}_1}/{\partial P}, \qquad\dot J = 0,\\\
        \dot P =  -\epsilon \,{\partial \bar{H}_1}/{\partial Q},
\end{array} \right.
\end{equation}
where $\bar H_1(Q, P,J):=\frac{1}{2\pi}\int_0^{2\pi} H_1(Q, P,J, \phi)\,d\phi$.
\end{definition}

Since there is no evolution of $J$ in the averaged system, one can  fix it and regard $J$ as a parameter
for the Hamiltonian system with $\ell$ degrees of freedom, where
$\bar{H}(Q,P)=\bar{H}_J(Q,P)=H_0(J)+\epsilon \bar{H}_1(Q,P,J)$.
\medskip

Let $(q,p, I)$ belong to a domain $D\subset \mathbb R^{2\ell+1}$, and $D_\delta\subset D$  stands for
a subdomain whose $\delta$-neighbourhood belongs to $D$. Assume that the Hamiltonian
$H$ is $C^3$-bounded for
$(q, p, I, \phi) \in D\times \mathbb T$, as well as ${\partial H_0(I)}/{\partial I}>C>0$ in $D$ and
$(Q(t), P(t), J(t))\in D_\delta$ for all $0\le t\le 1/\epsilon$.

\begin{theorem}\label{thm:Ave_Ham}{\rm (cf. \cite{ArnMM, arko})}
For sufficiently small positive $\epsilon$ (i.e. $0<\epsilon<\epsilon_0$) solutions of the actual system \eqref{eq:Ham_eq} and averaged system \eqref{eq:Ave_Ham_eq}
with the same initial conditions $(q(0), p(0), I(0))=(Q(0),P(0), J(0))$
remain $\epsilon$-close to each other for $0\le t\le 1/\epsilon$:
$|I(t)-I(0)|<C_0\epsilon$ and $|q(t)-Q(t)|+|p(t)-P(t)|<C_0\epsilon$, where $C_0$ does not depend on $\epsilon$.
\end{theorem}


Proof of this theorem consists of constructing a canonical transformation $\epsilon$-close to the identity and mapping the original system to the averaged one modulo $\epsilon^2$-terms. Then for the time
$0\le t\le 1/\epsilon$  solutions of the averaged system remain $\epsilon$-close to those of
the original one, see \cite{ArnMM}.


\begin{remark}
Consider now a fast-oscillating nonautonomous Hamiltonian system with ``$\ell$ and a half" degrees of freedom, whose Hamiltonian is $H=H(p,q,\omega t)$ with  high frequency $\omega=\mu/\epsilon$. The associated Hamiltonian equations are $\dot q={\partial H}/{\partial p}$ and $\dot p= -{\partial H}/{\partial q}$.
The fast variable $\phi=\omega t$ can be regarded as a new independent space variable by passing to the new autonomous Hamiltonian system for  $\widetilde H=\omega I+H(p,q,\phi)$
with $\ell+1$ degrees of freedom, where variable $I$ is conjugate to $\phi$.
Combined with the reparametrization $t\mapsto \tau=t/\epsilon$ this
leads to the system of the above type  (where now the upper dot stands for the derivative in the fast time $\tau$):
\begin{equation}
\left\{
  \begin{array}{l}
~ \dot q=\quad\epsilon\,{\partial H}/{\partial p}\,, \quad  \dot I= -\epsilon\,{\partial H}/{\partial \tau},\\\
 \dot p= -\epsilon\,{\partial H}/{\partial q}\,, \quad \dot \phi=~ \mu.
\end{array} \right.
\end{equation}
\end{remark}


\subsection{Averaging in natural systems}\label{sect:Ave_natHamiltonian}
In this section we study averaging in natural systems on the cotangent bundle of a principle circle bundle $\pi: Q\rightarrow B$. 
The general form of a natural Hamiltonian function on $T^*Q$ is 
$$
H(x,y)=\frac 12 (y,\,y)_x+U(x)
$$
for $(x,y)\in T^*_xQ$.

We start by considering natural systems on the cotangent bundle 
$T^*(\mathbb R^\ell\times  \mathbb T)$ of a direct product with Hamiltonians
\begin{equation}\label{eq:natHamiltonian}
H(q,\phi; p,\gamma)=\frac 12\left((p,\gamma), (p,\gamma)\right)_{(q,\phi)}+U(q,\phi),
\end{equation}
where $(q,\phi)\in\mathbb R^\ell\times  \mathbb T$ and $(p,\gamma)\in T_{(q,\phi)}^*(\mathbb R^\ell \times \mathbb T)$. 
Assume that the function $H(q,\phi; p,\gamma)$ is $2\pi$-periodic in $\phi$, 
$U(q,\phi)=U_0(q)+\epsilon U_1(q,\phi),$
 and the metric on $\mathbb R^\ell \times \mathbb T$ has the form
\begin{equation}\label{eq:metric}
\left( (p,\gamma), (p,\gamma) \right)_{(q,\phi)}=p\cdot p+2\gamma a(q,\phi)\cdot p+h(q,\phi)\,\gamma^2,
\end{equation}
where  dot  stands for the Euclidean inner product and where 
the corresponding coefficients $a$ and $h$ have the following expansions in $\epsilon$ as $\epsilon\to 0$:
$$
a(q,\phi)=a_0(q)+\epsilon a_1(q,\phi),\;\; h(q,\phi)=h_0(q)+\epsilon h_1(q,\phi),
$$
with  functions $a_1(q,\phi),\;h_1(q,\phi)$ and $U_1(q,\phi)$ of zero mean with respect to $\phi$.

The Hamiltonian equations for this Hamiltonian function $H(q,\phi; p,\gamma)$ and 
 symplectic structure $\omega=(1/\epsilon)\,dq\wedge dp+d\phi\wedge d\gamma$ 
on $T^*(\mathbb R^\ell\times \mathbb T)$ are
\begin{equation}\label{eq:natHamiltonian_eq}
\left\{
  \begin{array}{l}
				~ \dot q = \;\;\epsilon \,(p+\gamma a(q,\phi)),\\\
       \dot \phi =  a(q,\phi)\cdot p+h(q,\phi)\,\gamma, \\\
 			\dot p =  -\epsilon\, \partial (\gamma a(q,\phi)\cdot p+(1/2)\,h(q,\phi)\gamma^2 + U(q,\phi))/{\partial q}\,, \\\
        \dot \gamma = -\epsilon\, \partial (\gamma a_1(x,\phi)\cdot p+(1/2)\,h_1(q,\phi)\gamma^2 +U_1(q,\phi))/{\partial \phi}.
\end{array} \right.
\end{equation}
In these equations $\phi$ is the fast variable. 
Below we prove that the averaged Hamiltonian for (\ref{eq:natHamiltonian}) is
\begin{equation}\label{eq:Ave_natHamiltonian}
\bar{H}(Q,P,\mu)=\frac 12\,P\cdot P+\mu \,a_0(Q)\cdot P+\frac 12\,\mu^2 h_0(Q)+U_0(Q)\,,
\end{equation}
where the part $(1/2)\,\mu^2 h_0(Q)+U_0(Q)$ is related to  an effective potential, while the term $\mu \,a_0(Q)\cdot P$
linear in impulses  is related to a magnetic-gyroscopic-like force. Namely, one has the following statement,
which is an adaptation of Theorem \ref{thm:Ave_Ham} to the system \eqref{eq:natHamiltonian_eq}.

\begin{theorem}\label{thm:Ave_natHamiltonian}
For sufficiently small $\epsilon>0$  solutions for the original Hamiltonian  \eqref{eq:natHamiltonian} and the averaged Hamiltonian \eqref{eq:Ave_natHamiltonian}
with the same initial conditions $(q(0), p(0), \gamma(0))=(Q(0),P(0),\mu(0))$
remain $\epsilon$-close to each other for $0\le t\le 1/\epsilon$:
$|\gamma(t)-\mu(t)|+|q(t)-Q(t)|+|p(t)-P(t)|<C_0\epsilon$, where $C_0$ does not depend on $\epsilon$.
\end{theorem}
\begin{proof}
 The Hamiltonian equations \eqref{eq:natHamiltonian_eq} can be rewritten in the following form
\begin{equation}\label{eq:natHamiltonian_eq2}
\left\{
  \begin{array}{l}
				~ \dot q = \;\;\epsilon \,(p+\gamma a_0(q))+\epsilon^2 \gamma a_1(q,\phi)\,,\\\
       \dot \phi =  a(q,\phi)\cdot p+h(q,\phi)\,\gamma, \\\
 		\dot p =  -\epsilon\, \partial (\gamma a_0(q)\cdot p+(1/2)\,h_0(q)\gamma^2 + U_0(q))/{\partial q}\\\
	~\qquad\qquad\qquad\qquad \;\; -\epsilon^2 \partial (\gamma a_1(q,\phi)\cdot p+(1/2)\,h_1(q,\phi)\gamma^2 + U_1(q,\phi))/{\partial q}\,, \\\
        \dot \gamma = -\epsilon\, \partial (\gamma a_1(x,\phi)\cdot p+(1/2)\,h_1(q,\phi)\gamma^2 +U_1(q,\phi))/{\partial \phi}.
\end{array} \right.
\end{equation}
These equations differ by $\epsilon^2$-terms from those for  the averaged Hamiltonian \eqref{eq:Ave_natHamiltonian} in the symplectic structure $\Omega=dP\wedge dQ$:
\begin{equation}\label{eq:Ave_natHamiltonian_eq}
\left\{
  \begin{array}{l}
				~ \dot Q = \epsilon\, (P+\gamma a_0(Q))\,,\\\
 		\dot P =  -\epsilon\, \partial (\gamma a_0(Q)\cdot P+(1/2)\,h_0(Q)\gamma^2 + U_0(q))/{\partial Q}\,,\\\
       \dot \mu = 0,
\end{array} \right.
\end{equation}
so according to Theorem \ref{thm:Ave_Ham}, we obtain the required proximity of solutions for the original and averaged equations.
\end{proof}

\begin{remark} \label{rm:Ave_natHamiltonian}
In the shifted  coordinates $(Q, P)\to (Q, P_1:=P+\mu\,a_0(Q))$, 
the averaged Hamiltonian \eqref{eq:Ave_natHamiltonian} becomes
\begin{equation}\label{eq:Ave_natHamiltonian2}
\bar{H}(Q,P_1,\mu)=\frac 12\,P_1\cdot P_1+\frac 12\,\mu^2 (h_0(Q)-a_0(Q)\cdot a_0(Q))+U_0(Q),
\end{equation}
while the symplectic structure becomes
$$
\Omega=d(PdQ)=d((P_1-\mu\,a_0(Q))\,dQ)=dP_1\wedge dQ-\mu\,d(a_0(Q)\,dQ).
$$
One notices that  a new (effective) potential  now includes an additional term, 
$$
\bar U_{\mu}(Q):=\frac 12\,\mu^2 (h_0(Q)-a_0(Q)\cdot a_0(Q))+U_0(Q),
$$ 
while the new symplectic structure acquires the magnetic term $\mu\,d(a_0(Q)\,dQ)$.
\end{remark}

Now we are ready to prove 

\begin{corollary} 
The solutions of the reduced Hamiltonian  system and projections to slow manifold of solutions of the actual system with the same initial conditions remain $\epsilon$-close to each other for $0\le t\le 1/\epsilon$.
\end{corollary}

\begin{proof}
Comparing the result of Theorem \ref{thm:ave-red} with Remark \ref{rm:Ave_natHamiltonian} we observe that the Hamiltonian \eqref{aveHam} in the theorem coincides with the averaged Hamiltonian of the natural system obtained above.

As a matter of fact, the  consideration of Remark \ref{rm:Ave_natHamiltonian}
can be seen as a local (coordinate) version of the proof of Theorem  \ref{thm:ave-red}. Indeed, the averaged metric in local coordinates can be expressed as
$$
\overline{((p,\mu),\;(p,\mu))}_{(q,g)}^{\mathbb T}=p\cdot p+2\mu \,a_0(q)\cdot p+\mu^2 h_0(q),
$$
where $(q,g)\in B\times\mathbb T$ and $(p,\mu)\in T_{(q,g)}^*( B\times\mathbb T)$,  cf. (\ref{eq:Ave_natHamiltonian}). Therefore the corresponding mechanical connection $\bar{\mathcal{A}} \in \Omega^1(B,\mathbb R)$ in local coordinates is $\bar{\mathcal{A}}=a_0(q)\;dq$. One can see that the symplectic structure and effective potential in Remark \ref{rm:Ave_natHamiltonian} coincide with those in Theorem~\ref{thm:ave-red}.

Since one obtains the same averaged system both via the ``local" proof of Section \ref{sect:Ave_natHamiltonian} and via the ``global" proof of Theorem \ref{thm:ave-red}, then Theorem \ref{thm:Ave_natHamiltonian} guarantees the required  closeness of averaged  and original solutions.
\end{proof}

\begin{remark}
While Theorem \ref{thm:Ave_natHamiltonian} deals with the topologically trivial 
$\mathbb T$-bundle $\mathbb R^\ell\times \mathbb T$,
the  result on the existence of an averaged system holds for a topologically nontrivial bundle as well.
In the general case of a nontrivial $\mathbb T$-bundle $\pi:Q\to B$ 
we assume that the Hamiltonian system has fast motions along fibers and slow motions on the base $B$.
Under the assumption that  the  oscillatory parts of the Hamiltonian are of order $\epsilon$, i.e. 
\begin{equation}\label{eq:cond}
U(q)-\overline{U(q)}^{\mathbb T}\sim \;O(\epsilon) \;\;\text{and}\;\;(p,p)_q-\overline{(p,p)}_q^{\mathbb T}\sim\; O(\epsilon),
\end{equation}
one can  introduce the averaged kinetic energy
$(1/2) \overline{(p,p)}_q^{\mathbb T}$ and averaged potential energy $\overline{U(q)}^{\mathbb T}$ by
averaging the system along the fibers, see Section \ref{sect:averaging}. Then one proves the averaging theorem
by  considering only the local picture of the principal bundle $\pi:Q\to B$,
since any bundle is locally trivial. 
\end{remark}


\subsection{Examples of averaged systems}\label{subsect:examples}
\begin{example}[\textrm{\textit{(A pendulum with rapidly oscillating suspension point)}}]\label{ex:pendulum}
Consider the motion of a pendulum with vertically vibrating suspension point. Set $\theta$ to be the angle of
deviation of the pendulum from the vertical, $a$ and $\omega$ are the amplitude and frequency of the
oscillation of the suspension point, $l$ is the length of the pendulum and $g$ is the acceleration of gravity.
We assume that the amplitude $\epsilon a$ is of order $\epsilon$ and the frequency $\omega={\mu}/{\epsilon}$ is of order ${1}/{\epsilon}$, i.e. the suspension point oscillates with high frequency and small amplitude.
The corresponding potential  is $U(\theta)=-gl \cos\theta$ and the 
Hamiltonian function is
$$
H(\theta,p,t)=\frac 12 \left( \frac pl-a\mu \sin\omega t \sin\theta\right)^2-gl \cos\theta\,.
$$
This system differs from a natural one because of the shift in $p$, and 
the standard reasoning goes as follows, see e.g. \cite{arko}.
 Let $\phi=\omega t$ be the fast variable. In order to get rid of the $\phi$-dependence in the Hamiltonian 
 of vibrating pendulum, we seek for a canonical transformation $(\theta,p)\mapsto(\theta_1,p_1)$ 
 with a generating function $p_1\theta+\epsilon S_1(\theta,p_1,\phi)$, 
 where the function  $S_1$ is $2\pi$ periodic in $\phi$. Then the new Hamiltonian becomes
$$
\mathscr{H}(\theta_1,p_1,\phi)=\mu \frac{\partial S_1}{\partial \phi}+H(\theta_1,p_1,\phi)+O(\epsilon),
$$
as $\epsilon\to 0$. By taking $
S_1(\theta,p_1,\phi)=-( p_1/l)\, a \cos\phi \sin\theta+( 1/8)\, a^2\mu \sin^2\theta \sin 2\phi,
$ we  obtain the following Hamiltonian averaged to the first order in $\epsilon$:
$$
\mathscr{H}(\theta_1,p_1,\phi)=\frac {p_1^2}{2l^2}-gl\cos\theta_1+\frac 14 a^2\mu^2\sin^2\theta_1+O(\epsilon).
$$
Notice the appearance of an additional positive definite quadratic term in the effective potential
$$
U_{\mu}(\theta_1)=-gl\cos\theta_1+( 1/4) a^2\mu^2\sin^2\theta_1\,.
$$
It causes such an interesting dynamical phenomenon as the stability of the upper position of the pendulum.
\end{example}

\begin{example}[\textrm{\textit{(A particle in a rapidly oscillating  potential)}}]\label{ex:rapid}
Consider the motion of a particle in a rapidly oscillating potential, following \cite{cole}. The corresponding
Hamiltonian function is $H(q,p, t)=( 1/2) \,p\cdot p+U(q,  t/{\epsilon})$, where the potential function $U(q,\tau)$
is $2\pi$-periodic with respect to $\tau$.
To obtain an averaged Hamiltonian modulo the third order  in $\epsilon$ one needs to iteratively
apply canonical transformations $(q,p)\to (Q,P)$  four times (see \cite{cole} for details). The results is
$$
\bar{H}(Q,P)=\frac 12\, P\cdot P+\bar U(Q)+\frac{\epsilon^2}{2}\overline{V'\cdot V'}-\epsilon^3\,\overline{S''V'}P,
$$
where $\bar U(q)=\frac 1 {2\pi}\int_0^{2\pi} U(q,\tau)\;d\tau$ is the time 
average of $U$ over one temporal period,  functions $V$ and $S$ stand for temporal antiderivatives $V(q,\tau):=\int^{\tau} [U(q,\theta)-\bar U(q)]\,d\theta, \;S(q,\tau):=\int^{\tau} V(q,\theta)\, d\theta$, with the
constants of integration   chosen such that $\bar V=\bar S=0$,
and the  prime denotes the derivative with respect to the new space variable $Q$.
\end{example}

\begin{example}[\textrm{\textit{(Foucault pendulum)}}]
By considering small and rapid oscillations of an ideal pendulum on a  sphere rotating with angular velocity $\Omega$, one can observe that the plane of oscillation will be rotating with the angular velocity $-\Omega\sin\lambda$, where $\lambda$ is the pendulum latitude. This is an example of the Foucault pendulum, see \cite{ArnMM} for a detailed description.
More generally, one can consider a curved surface and a pendulum slowly transported along a path on the surface.
In this case the plane of oscillation turns out to be parallel transported along the path on the surface, see the discussion 
in \cite{Klein}, and it can be regarded as a physical interpretation of the Levi-Civita connection.

One can use the averaging-reduction theory to interpret this phenomenon. First, notice that this system 
has the same $\mathbb T$-bundle structure as the one in the gyroscope Example \ref{ex:spinDisk}. Hence by shifting the momentum similarly to that in Theorem \ref{thm:cotbundle}, one obtains a Hamiltonian system on (the cotangent bundle of) the tangent bundle of the surface with a twisted symplectic structure, where the twist is given by an additional curvature term related to the Levi-Civita connection. Then the averaging-reduction procedure allows one to descend this system to (the cotangent bundle of) the unit tangent bundle of the surface, on which the parallel transport is observed. We plan to discuss this example in detail elsewhere.
\end{example}


\medskip

\section{Applications}\label{sect:applications}

\subsection{Pendulum with a vibrating suspension point}
In Section \ref{subsect:examples}, we  obtained the averaged Hamiltonian for a pendulum with a vibrating suspension point using the classical averaging method, and  observed the appearance of an additional quadratic term in the effective potential. In this section,  using the averaging-reduction procedure of Section \ref{sect:averaging}, we show that this additional term is the result of symplectic reduction.
\smallskip

We start with the following  Hamiltonian function describing a natural mechanical system with a rapidly oscillating potential:
\begin{equation}\label{eq:Ham_Pen1}
 H(x,p, t)=\frac 12 (p,\,p)+U(x)+\epsilon\,\omega^2\;\widetilde{U}(x,\omega t),
\end{equation}
where the frequency $\omega=\mu/\epsilon$ is of order $1/{\epsilon}$, and the oscillating part of the potential $\widetilde{U}(x,\phi)$ is $2\pi$-periodic and has zero mean with respect to $\phi$.

Introduce the fast time $\tau=t/\epsilon$ and fast variable $\phi=\omega t=\mu \tau$. We split any (vector-) function $f=f(t,\phi)$ depending
on two times $t$ and $\tau$ (and $2\pi$-periodic in $\phi=\mu\tau$)   into the mean and oscillatory parts:
$$
f(t,\phi)=\overline{f}(t)+\widetilde{f}(t,\phi),
$$
where $\overline{f}=({1}/{2\pi})\,\int_{0}^{2\pi}f(t,\phi)\,d\phi$.
Now regard the fast variable $\phi$ as a new coordinate. Note that when the fast time $\tau$ changes by 1,
the slow time $t$  changes only by $\epsilon$. So for a motion $x(t,\phi)=\overline{x}(t)+\widetilde{x}(t,\phi)$
described by the Hamiltonian function (\ref{eq:Ham_Pen1}), one can fix $\overline{x}(t)$ and regard
$\widetilde{x}(t,\phi)=\widetilde{x}(\overline{x},\phi)$ as a function of $\overline{x}$ and $\phi$ 
modulo $O(\epsilon)$ as $\epsilon\to 0$. 
In other words, one can consider  a map from a suspension over $M$, a manifold $M\times\mathbb T$,  to $M$ itself,
where  a point $(\overline{x},\phi)\in M\times \mathbb T$ is mapped to $x=\overline{x}(t)+\widetilde{x}(\overline{x},\phi)\in M$. Here $\widetilde{x}(\overline{x},\phi)$ can be obtained by solving the Hamiltonian system corresponding to the above Hamiltonian  function $H$ in variables $(\widetilde x, \widetilde p)$:
\begin{equation}\label{eq:Ham_fib}
 H(\widetilde{x},\widetilde{p}, t)=\frac 12 (\widetilde{p},\widetilde{p})+U(\overline x)+\epsilon\,\omega^2 \;\widetilde U(\overline{x}+\widetilde{x},\omega t),
\end{equation}
with the initial conditions such that the solution $(\widetilde x(\bar x,\phi), \widetilde p(\bar x,\phi))$ has zero mean value with respect to $\phi$.

In order to compute the metric on the suspension manifold $M\times \mathbb T$ we set  $\mu=1$ (i.e. $\omega=1/\epsilon$  and $\phi=\tau$).  Thinking of   position $\widetilde{x}$ 
and momentum  $\widetilde{p}$ as depending on the average position $\overline{x}$ and fast time $\tau$, 
we denote by  $(\widetilde{x}(\overline{x},\tau),\widetilde{p}(\overline{x},\tau))$ 
the solution of the corresponding Hamiltonian system with initial conditions corresponding to zero mean value  relative  to $\tau$.
\smallskip

According to the above consideration, the configuration space is  
a principal $\mathbb T$-bundle $\pi:M\times \mathbb T\rightarrow M$, where
the $\mathbb T$-action on $Q=M\times \mathbb T$ is the shift in fibres
$\phi\circ(\overline{x},e^{i\phi'})=(\overline{x},e^{i(\phi+\phi')}).$
Applying the averaging-reduction theory of Section \ref{sect:averaging}, we obtain the following statement.

\begin{theorem}\label{thm:Pen_thm1}
The Hamiltonian system (\ref{eq:Ham_Pen1}) averaged using the standard averaging method in Example \ref{ex:pendulum} coincides with the system obtained
on the reduced symplectic manifold $(T^*M,\omega_{can})$ with the canonical symplectic structure, the Hamiltonian function
\begin{equation}\label{eq:avePen1}
H_{slow}(\overline{x},\overline{p})=\frac 12 (\overline{p},\overline{p})+U_{\mu}(\overline{x}),
\end{equation}
for $(\overline{x},\overline{p})\in T^*M$, and with the effective potential 
$$
U_{\mu}(\overline{x})=U(\overline{x})+\frac{\epsilon^2\omega^2}{4\pi}\,\int_0^{2\pi}(\widetilde{v}(\overline{x},\tau),\widetilde{v}(\overline{x},\tau))\,d\tau\,.
$$
\end{theorem}
\begin{proof}
Now the averaged metric is given by
$$
((v,\gamma),(v,\gamma))_{(\overline{x},\phi)}=(v,v)+2\pi\,\gamma^2\left(\int_0^{2\pi}(\widetilde{v}(\overline{x},\tau),\widetilde{v}(\overline{x},\tau))\,d\tau\right)^{-1},
$$
where $(v,\gamma)\in T_{\overline{x}}M\times T_{\phi}\mathbb T$. Also, note that the value of $({1}/{2\pi})\int_0^{2\pi}(\widetilde{v}(\overline{x},\tau),\widetilde{v}(\overline{x},\tau))\,d\tau$  depends on $\overline{x}$ only.

The corresponding fiber inertia operator $\mathbb{I}(\overline{x}):\mathfrak t=\mathbb{R}\rightarrow\mathfrak t^*=\mathbb{R}$
at $\overline{x}\in M$  is given by 
$\mathbb{I}(\overline{x})\gamma
=2\pi\gamma\left(\int_0^{2\pi}(\widetilde{v}(\overline{x},\tau),\widetilde{v}(\overline{x},\tau))\,d\tau\right)^{-1}$, while
the momentum map $J:T^*(M\times \mathbb T)\rightarrow\mathfrak t^*=\mathbb{R}$ is given by
$J(\bar x,\phi,p,\eta)=\eta$.
The averaged connection $\widebar{\mathcal A}\in\Omega^1(M\times \mathbb T,\mathbb R)$ 
on the principal trivial $\mathbb T$-bundle $M\times \mathbb T$ is given by
$\widebar{\mathcal A}(\bar x,\phi,v,\gamma)=\gamma.$ This connection is flat, $\widebar{\mathcal A}=d\phi$.
\smallskip

Recall that the magnetic term is proportional to the curvature of the bundle $\pi:Q\to M$.
The flatness of $\widebar{\mathcal A}$ implies that the reduced symplectic structure on the  manifold
$T^*M=J^{-1}(\mu)/\mathbb T$ has no magnetic term, i.e. it coincides with the canonical symplectic structure $\omega_{can}$.
By taking $\mu=d\phi/d\tau=\epsilon\omega$ in
Theorem \ref{thm:ave-red}, the effective potential $U_{\mu}$ in the Hamiltonian function (\ref{eq:avePen1}) is
$$
U_{\mu}(\overline{x})=U(\overline{x})+1/2\,\langle \mu,\mathbb I(\overline{x})^{-1}\mu\rangle
=U(\overline{x})+({\mu^2}/{4\pi})\int_0^{2\pi}(\widetilde{v}(\overline{x},\tau),\widetilde{v}(\overline{x},\tau))\,d\tau
$$
$$
=U(\overline{x})+\epsilon^2({\omega^2}/{4\pi})\int_0^{2\pi}(\widetilde{v}(\overline{x},\tau),\widetilde{v}(\overline{x},\tau))\,d\tau.
$$
\end{proof}

Now return to Example \ref{ex:pendulum} in Section \ref{subsect:examples}. The Hamiltonian of a pendulum with a vibrating suspension point can be rewritten in the following form:
\begin{equation}\label{eq:Ham_Pen2}
 H(\theta,p)=\frac 12 \left(\frac pl\right)^2-(g-\epsilon\, a\, \omega^2\sin\tau)\,l\cos\theta.
\end{equation}

Recall our assumption  that the amplitude $\epsilon a$ is of order $\epsilon$ 
and the frequency $\omega=\mu/\epsilon$ is of order
$1/{\epsilon}$, which means that  the above Hamiltonian has the   form  (\ref{eq:Ham_Pen1}).
Hence we obtain the following result on its slow motion.

\begin{theorem}\label{thm:Pen_thm2}
The averaged Hamiltonian function reduces to the following Hamiltonian function describing the slow motion:
\begin{equation}\label{eq:avePen2}
H_{slow}(\overline{\theta}, \overline{p})=\frac 12 \left(\frac{\overline{p}}{l}\right)^2+U_{\mu}(\overline{\theta}),
\end{equation}
where the effective potential is $U_{\mu}(\overline{\theta})=(1/4)\, \mu^2 a^2\sin^2\overline{\theta}-gl\cos\overline{\theta}$.
\end{theorem}

\begin{proof}
Let $\tau=t/\epsilon$ and $(\widetilde{\theta}(\overline{\theta},\tau), \widetilde{p}(\overline{\theta},\tau))$ be the solution of Hamiltonian system corresponding to the function
$$
 \widetilde{H}(\widetilde{\theta}, \widetilde{p})=\frac 12 \left(\frac {\widetilde{p}}{l}\right)^2-(g- \frac a {\epsilon} \sin\tau)\,l\cos(\overline{\theta}+\widetilde{\theta}),
$$
with the initial values determined by the zero mean conditions of $(\widetilde{\theta}(\overline{\theta},\tau), \widetilde{p}(\overline{\theta},\tau))$ with respect to $\tau$.

Omitting  terms of order $\epsilon^2$ we obtain the following Newton equation for $\widetilde{\theta}$:
$$
 d^2\widetilde{\theta}/d t^2=-(a/(\epsilon l))\,{\sin\tau\sin\overline{\theta}}\,.
$$
By rewriting it for the fast time $\tau=t/\epsilon$ and integrating we obtain $d\widetilde{\theta}/d\tau=(\epsilon a/l)\,{\cos\tau\sin\overline{\theta}}$ and $\widetilde{\theta}=(\epsilon a/l)\,{\sin\tau\sin\overline{\theta}}$.
(In this integration one regards the right-hand side as a function of $\tau$ modulo higher order terms in $\epsilon$,
and uses the zero mean condition on  $d\widetilde{\theta}/d\tau$ and $\widetilde{\theta}$.)

Furthermore, the momentum map
$J:T^*(\mathbb T_{\bar\theta}\times \mathbb T_\phi)\rightarrow \mathbb{R}$, corresponding to the $\mathbb T$-action on $\mathbb T_{\bar\theta}\times \mathbb T_\phi$ equipped with the averaged metric is $J(\bar\theta,\phi,p,\mu)=\mu$, and $({l^2}/{2\pi})\int_0^{2\pi}(d\widetilde{\theta}/d\tau,d\widetilde{\theta}/d\tau)\,d\tau=\epsilon^2(a^2/2) \sin^2\overline{\theta}$.

 Therefore, by Theorem \ref{thm:Pen_thm1}, the reduced (or slow) symplectic manifold is $(T^*\mathbb T_{\bar\theta},\omega)$ with symplectic structure  $\omega=d\bar\theta\wedge d\bar p$  and
 the Hamiltonian of the slow motion is
$$
H_{slow}(\overline{\theta}, \overline{p})=\frac 12 \left(\frac{\overline{p}}{l}\right)^2+U_{\mu}(\overline{\theta}),
$$
where  the effective potential is
$$
U_{\mu}(\overline{\theta})=\frac{\omega^2 l^2}{4\pi}\int_0^{2\pi}(d\widetilde{\theta}/d\tau,d\widetilde{\theta}/d\tau)\,d\tau+U(\bar\theta)
$$
$$
=\epsilon^2\frac{\omega^2 a^2}{4}\sin^2\overline{\theta}-gl\cos\overline{\theta}=\frac 14\, \mu^2 a^2\sin^2\overline{\theta}-gl\cos\overline{\theta}\,.
$$
\end{proof}


\subsection{Craik-Leibovich equation}\label{sect:CL}
Consider the motion of an ideal fluid confined to a three-dimensional domain $D$ 
with  fast oscillating boundary  $\partial D$. The   averaged fluid motion is described by the following Craik-Leibovich (CL) equation on the fluid velocity field $v$:
\begin{equation}\label{eq:CraLei}
\left\{
\begin{array}{ll}
\frac{\partial v}{\partial t}+(v,\nabla)v+\mathrm{curl}\;v\times V_s=-\nabla p,\\
(v+V_s)\,||\,\partial D,\\
\mathrm{div}\,v=0,
\end{array}
\right.
\end{equation}
where $V_s$ is a (time-dependent) prescribed Stokes drift velocity related to the average of surface waves.
We refer to \cite{vla} for a derivation of the CL equations via perturbation theory.  
In \cite{yang1, yang2} the Hamiltonian structure of the CL equation was studied, along with 
a generalization of the perturbation theory to a principal $\mathbb T$-bundle over any group $G$ and derivation of the Euler equation associated with a certain central extension of  $G$.

In a more general setting, 
let $\mathrm{SDiff}(D)$ be the group of all volume-preserving diffeomorphisms of an $n$-dimensional Riemannian manifold $D$ with boundary $\partial D$. Its Lie algebra $\mathfrak g=\text{SVect}(D)$ consists of
all the divergence-free vector fields in $D$ tangent to the boundary $\partial D$, while the regular dual space
$\mathfrak g^*=\Omega^1(D)/d\Omega^0(D)$ of this Lie algebra is the space of cosets of 1-forms on $D$ modulo  exact 1-forms.

\begin{theorem}{\rm (see \cite{yang2})}\label{thm:genCL}
The  $n$-dimensional  Craik-Leibovich (CL) equation on the  space
$\Omega^1(D)/d\Omega^0(D)$ has the form
\begin{equation}\label{eq:genCL}
 \frac{d}{dt}\;[u]=-\mathcal{L}_{v+V_s}\;[u],
\end{equation}
where $v+V_s\in {\rm SVect}(D)$, and $[u]=[v^\flat]\in\Omega^1(D)/d\Omega^0(D)$ is the coset of the 1-form
$v^\flat$ metric-related to the vector field $v$ on $D$.
\end{theorem}

\begin{remark}
The requirement $v+V_s\in {\rm SVect}(D)$ provides the boundary condition of tangency of   the vector field $v+V_s$  to the boundary $\partial D$. Although the divergence-free vector fields $v$ and $V_s$ are not necessarily tangent to the boundary, the 1-forms $v^\flat$ and $V_s^\flat$, considered up to the differential of a function, are well-defined elements in the dual space $\Omega^1(D)/d\Omega^0(D)$ of Lie algebra ${\rm SVect}(D)$.
\end{remark}

Upon shifting the origin in  $\mathfrak g^*=\Omega^1(D)/d\Omega^0(D)$ by $[V_s^\flat]$,
equation (\ref{eq:genCL}) becomes
\begin{equation}\label{eq:genCL2}
\frac{d}{dt}\;[w]=-\mathcal{L}_{\mathbb{I}^{-1}[w]}\;\left([w]-[V_s^\flat]\right)
\end{equation}
for $[w]:=[u+V_s^\flat]$.

\begin{theorem}{\rm (see \cite{yang2})}\label{thm:cenExt}
The equation (\ref{eq:genCL2}) is the Euler equation on the central extension $\widehat{\mathfrak g}$ of the Lie algebra $\mathfrak{g}={\rm SVect}(D)$ by means of the 2-cocycle
$$
{\sigma}_{V_s}(X,Y):=-\left\langle\mathcal{L}_{X}\;V_s^\flat,Y\right\rangle
$$
associated with the vector field $V_s$.
\end{theorem}

\begin{remark}
More generally, for a  ``shift 2-cocycle'' 
$$
{\sigma}_{V_s}(X,Y):=-\left\langle{\rm ad}_X^*\mathbb I(V_s),Y\right\rangle=-\left\langle\mathbb I(V_s),[X,Y]\right\rangle
$$ 
on an arbitrary Lie algebra $\mathfrak g$,
consider the corresponding (possibly trivial) central extension $\widehat{\mathfrak g}$ 
of Lie algebra $\mathfrak g$ by means of this 2-cocycle.  
Then the Euler equation on the centrally extended Lie algebra $\widehat{\mathfrak g}$ is
\begin{equation}\label{eq:cenExtEuler}
\frac{d}{dt}\;\xi=-{\rm ad}^*_{\mathbb{I}^{-1}\xi}\left(\xi-\mathbb I(V_s)\right),
\end{equation}
where $\xi\in \mathfrak g^*$.
Applying this to the Lie algebra ${\rm SVect}(D)$, where ${\rm ad}_X^*\xi=\mathcal{L}_X \xi$, we obtain the above theorem.
\end{remark}

The averaging-reduction procedure gives us an explanation of the structure of central extension for 
the CL equation, which  appeared as the result of averaging, as well as the origin of the vector field $V_s$. More specifically, by applying the symplectic reduction procedure to a natural system on the principal 
$\mathbb T$-bundle $G\times \mathbb T\to G$ for a group $G$ (in the case of  an oscillating flow, $G$ is the group
$\text{SDiff}(D)$ of volume-preserving diffeomorphisms of $D$) 
one obtains the averaged system on the cotangent bundle $T^*G$ of the base $G$ with  
the symplectic structure endowed with a magnetic term.
Furthermore,  Theorem \ref{thm:poisson} provides the corresponding central extension.
First we prove the averaging-reduction theorem for a trivial bundle $G\times \mathbb T\to G$, and then explain the necessary changes in the general case. 

\begin{theorem}\label{thm:aveCL}
The averaging of a natural $G$-invariant Hamiltonian system on $T^*(G\times \mathbb T)$ reduces to a slow Hamiltonian system on the reduced symplectic manifold $(T^*G,\omega_\mu)$ with the Hamiltonian function $H_{slow}(\bar g,\bar r)=\frac 12 (\bar r,\,\bar r),$
where $(\bar g,\bar r)\in T^*G$, and the symplectic structure $\omega_\mu$  given by
$$
\omega_\mu=\omega_{can}-\mu\,\pi_G^*\,d\widetilde{\alpha},
$$
for the canonical symplectic form   $\omega_{can}$ on $T^*G$. Here $\pi_G:T^*G\rightarrow G$ is 
the cotangent bundle projection and $\widetilde{\alpha}$ is a certain 1-form on the base $G$ depending on the averaged metric.
\end{theorem}

\begin{proof}
For a point $q=(g,x)\in G\times \mathbb T$ let $p=(\nu, y)$ be  a covector at that point. 
A natural $G$-invariant  Hamiltonian system on $T^*(G\times \mathbb T)$ with Hamiltonian 
$H(q,p)=\frac 12 (p,\,p)+U(x)$ for a $G$-invariant metric on $G\times\mathbb T$  has the form:
$$
 H(g,x;\nu,y)=\frac 12 ((\nu,y),\;(\nu,y))_{(g, x)}+U(x).
$$
To write it more explicitly,\footnote{This is based on the same consideration as Remark \ref{rem:metric}.} 
let $\mathcal B^*(x): \mathfrak g^*\to T_x^*\mathbb T$ be the linear map associated with the $G$-invariant 
metric,\footnote{As a matter of fact, $\mathcal B^*$ will turn out to be dual of a  flat connection $\mathcal B$  in the bundle 
$G\times \mathbb T\to \mathbb T$, see Remark~\ref{rem:nontriv-bundle}.} and identify covector components $\nu$ with elements of $\mathfrak g^*$  
by right translations. Then the Hamiltonian  can be rewritten as
\begin{equation}\label{eq:Ham_CL2}
 H(g,x;\nu,y)=\frac 12 (y,\;y)_x+(y,\;\mathcal B^*\nu)_x+\frac 12 (\nu,\;\nu)_x+U(x).
\end{equation}

Now  fix a position $g$ of the slow motion and  consider the fast motion (i.e. $x$-dependence)
of the system. In this setting
we have $\nu=0$, since the fast motion has zero mean, omitting  higher order terms in $\epsilon$. Then the Hamiltonian for the fast motion becomes
\begin{equation}\label{eq:Ham_CL2_fast}
 \widetilde{H}(\widetilde x,\widetilde y)=\frac 12 (\widetilde y,\widetilde y)_{\widetilde x}+U(\widetilde x)\,.
\end{equation}
Note that the  Hamiltonian of fast motion does not depend on the group element $g$.

Suppose that the period of this motion is $2\pi/\omega$ and set the fast variable to be $\phi=\omega t$, where the frequency $\omega=\mu/\epsilon$.
Let $(\widetilde{x}(\phi),\widetilde{y}(\phi))$ be any Hamiltonian trajectory for  Hamiltonian function (\ref{eq:Ham_CL2_fast}) satisfying the following two conditions:  it is $2\pi$-periodic with respect to the fast variable $\phi$ and its initial condition is chosen in  a way to provide zero mean for $(\widetilde{x}(\phi),\widetilde{y}(\phi))$. 
Such a vanishing condition generically defines a 1-parameter family of solutions parametrized by the energy level
(e.g. such solutions differ by scaling for a quadratic potential $U$). 
We use the fast  variable $\phi$ to define the 
$\mathbb T$-action on $G\times \mathbb T$,  which is given by $\phi\circ(g,e^{i\phi'})=(g,e^{i(\phi+\phi')})$. 

Since  the  Hamiltonian of fast motion does not depend on the group element $g$,  the averaged kinetic energy $(1/2\pi)\int_0^{2\pi}(\widetilde{v}(\phi),\widetilde{v}(\phi))\,d\phi$
does not depend on $g$ as well, where $\widetilde v $ is related to $\widetilde y$ by means of the metric: $\widetilde v^\flat=\widetilde y$. Take the solution $(\widetilde{x}(\phi),\widetilde{y}(\phi))$ for which 
 the averaged kinetic energy $(1/2\pi)\int_0^{2\pi}(\widetilde{v}(\phi),\widetilde{v}(\phi))\,d\phi$ is 
 $ {1}/{\epsilon^2}$.

 Then the averaged metric $\overline{(\cdot\,,\cdot)}^{\mathbb T}$ on $G\times\mathbb T$ 
 (cf. Remark \ref{rem:metric}) is given by
\begin{equation}
\begin{array}{rcl}
\overline{((v,\gamma),(v,\gamma))}^{\mathbb T}&=&(v,v)+2\pi(2 \gamma\,\langle \widetilde{\alpha},v\rangle+\gamma^2)
\left(\int_0^{2\pi}(\widetilde{v}(\phi),\widetilde{v}(\phi))\,d\phi \right)^{-1}\\
&=&(v,v)+\epsilon^2 \,(2\gamma\,\langle \widetilde{\alpha},v\rangle+ \gamma^2)\,,
\end{array}
\end{equation}
where $(v,\gamma)\in T_{g}G\times T_{\phi}\mathbb T\simeq T_{g}G\times \mathbb R$, while the
1-form $\widetilde{\alpha}\in T^*G$ depends on  $\widetilde y$. This 1-form is right-invariant,  i.e. it satisfies
$\widetilde{\alpha}(hg)=R_{g^{-1}}^*\widetilde{\alpha}(h)$. Indeed,
the metric on $G\times \mathbb T$ is right-invariant under the $G$-action, hence its average with respect to
the $\mathbb T$-action is right-invariant under the $G$-action as well.

The averaged inertia operator $\mathbb{I}(g):\mathfrak t=\mathbb{R}\rightarrow\mathfrak t^*=\mathbb{R}$ 
is $\mathbb{I}(g)\gamma=\epsilon^2\gamma$, while the
averaged momentum map $ J:T(G\times \mathbb T)\rightarrow\mathfrak t^*=\mathbb{R}$ is given by
 $ J(g,\phi, v,\gamma)=\langle \epsilon^2 \widetilde{\alpha},v\rangle+\epsilon^2\gamma$, 
 cf. Proposition \ref{prop:momentum}.
Finally, the averaged connection $\bar{\mathcal A}\in\Omega^1(G\times \mathbb T,\mathbb R)$ on the principal
$\mathbb T$-bundle $G\times \mathbb T$ is given by $\bar{\mathcal A}(g,\phi)= \widetilde{\alpha}+d\phi.$

\smallskip

According to Theorem \ref{thm:ave-red}  the magnetic term in the symplectic structure is $\beta_\mu=\mu\,\pi_G^*d \bar{\mathcal A}=\mu\,\pi_G^*d \widetilde{\alpha}$ for the corresponding 
$ \mu\in\mathbb R$. Therefore, the  symplectic structure  on the reduced manifold $T^*G\cong  J^{-1}(\mu)/\mathbb T$ 
is the 2-form $\omega_\mu=\omega_{can}-\mu\,\pi_G^*d\widetilde{\alpha}$. The correction in the effective potential $U_{\mu}$ of   Hamiltonian function (\ref{aveHam}) is  constant,
and hence   can be omitted.
\end{proof}

\begin{remark}\label{rem:nontriv-bundle}
For the general case of a topologically nontrivial $G$-bundle $\pi:\,M\to N$, consider an open subset 
$\mathcal O\subset N$. Then  locally in the base  trivialize $M|_{\pi^{-1}(\mathcal O)}\cong \mathcal O\times G$ and the cotangent bundle $T^*M|_{\pi^{-1}(\mathcal O)}\cong T^*\mathcal O\times G \times \mathfrak g^*$. This  decomposition of the tangent/cotangent spaces gives us a flat connection $\mathcal B\in\Omega^1(\mathcal O,\mathfrak g)$ on $ \mathcal O\times G$. Therefore, for the local representation $(x,g; y,\nu)$ of $(q,p)$, we have
$q=(x,g)$ and $p=(y+\mathcal B^*\nu,\nu)$, and then the Hamiltonian $ H(q,p)=\frac 12 (p,\,p)+U(\pi(q))$
in $T^*M$  becomes
$$
\begin{array}{rcl}
H(x,g;y,\nu)&=&\frac 12 (y+\mathcal B^*\nu,y+\mathcal B^*\nu)_x+\frac 12 (\nu,\nu)_x+U(x)\\
                    &=&\frac 12 (y,y)_x+(\mathcal B^*\nu,y)_x+\frac 12 (\mathcal B^*\nu,\mathcal B^*\nu)_x+\frac 12 (\nu,\nu)_x+U(x).
\end{array}
$$
By combining the third and forth terms in the above expression of the general Hamiltonian and by switching the order of coordinates, we obtain the Hamiltonian (\ref{eq:Ham_CL2}) in the proof. Finally, by confining ourselves to $x$-periodic solutions in the base lying inside 
$ \mathcal O$ one can reduce the setting to a Hamiltonian on $T^*(G\times \mathbb T)$.
\end{remark}

\begin{remark}
For the oscillating flow we
consider the   principal $G$-bundle $\pi:M\rightarrow N$ with $G=\text{SDiff}(D)$ to be the group of
volume-preserving diffeomorphisms. The infinite-dimensional manifold $M$ is the space of all
volume-preserving embeddings of the reference manifold $D$ to $\mathbb R^n$, while $N$ is the manifold of all  boundaries of such embeddings, i.e. hypersurfaces  diffeomorphic to $\partial D$ and bounding diffeomorphic manifolds with the same volumes.

For the averaged metric, the 1-form $\widetilde{\alpha}$ on $G$ is metric-related to a vector field $V_s$ on $G$  via $\widetilde{\alpha}=V_s^\flat$. This defines the 2-cocycle
${\sigma}_{V_s}(X,Y):=-\left\langle\mathcal{L}_{X}\;V_s^\flat,Y\right\rangle$ on the Lie algebra ${\rm SVect}(D)$ in Theorem \ref{thm:cenExt}.
\end{remark}

\smallskip

\subsection{Particles in  rapidly oscillating potentials}

It turns out that the motion of a particle in a rapidly oscillating potential field  \cite{cole} can also be viewed in the
context of the symplectic averaging-reduction setting. Moreover, in this example one observes
both phenomena: additional terms in the effective potential and magnetic correction to the symplectic structure.

Namely, consider the Hamiltonian function
\begin{equation}\label{eq:Ham_Par}
H(x,p, \omega t)=\frac 12 \,p\cdot p+\epsilon^2\mu^2\;U(x, \omega t),
\end{equation}
where $x\in \mathbb R^n$, the potential function $U(x,\phi)$ is $2\pi$-periodic with respect to (the fast variable) $\phi$ and the
frequency $\omega= \mu/\epsilon$ is of order $1/\epsilon$.
By ingenious repeated application of  canonical transformations (see \cite{cole} for details, where $\mu=1$) one obtains
an averaged Hamiltonian up to the third order in $\epsilon$:
\begin{equation}\label{eq:Ham_ave}
\overline{H}(\overline x,\overline p)=\frac 12 \overline p\cdot\overline p+\overline U +\frac{\epsilon^2\mu^2}{2}\overline{V'\cdot V'}-\epsilon^3\mu\;\overline{S''V'}\,\overline p,
\end{equation}
where $\overline U(x)=\frac 1 {2\pi}\int_0^{2\pi} U(x,\tau)\;d\tau,\;V(x,\tau)=\int^{\tau} U(x,\theta)-\overline U(x)\; d\theta,$ and 
$ S(x,\tau)=\int^{\tau} V(x,\theta)\; d\theta$ with the constant of integration chosen so
that $\overline V=\overline S=0$, and where prime denotes the derivative with respect to the space variable $x$.

On the other hand, the symplectic averaging-reduction procedure applied to the Hamiltonian (\ref{eq:Ham_Par})
is as follows. Regard  the fast variable $\phi\in \mathbb T$ as a new (periodic) coordinate, while the quotient along 
$\phi$-fibers is the slow manifold with coordinates $\overline x\in \mathbb R^n$,  mean values of the solutions.
The $\mathbb T$-action on the principal $\mathbb T$-bundle $\pi:\mathbb R^n\times \mathbb T\rightarrow \mathbb R^n$
is given by $\phi\circ(\overline{x},e^{i\phi'})=(\overline{x},e^{i(\phi+\phi')})$.

\begin{theorem}\label{thm:Par_thm}
The averaged Hamiltonian system \eqref{eq:Ham_ave} for the natural system    \eqref{eq:Ham_Par}
 is equivalent to the result of the symplectic averaging-reduction
procedure, i.e. the Hamiltonian system
on the reduced symplectic manifold $(T^*\mathbb R^n,\omega_{\mu})$ with the Hamiltonian function:
\begin{equation}\label{eq:avePar}
H_{slow}(\overline{x},\overline{p})=\frac 12 \,\overline{p}\cdot\overline{p}+U_{\mu}(\overline{x}),
\end{equation}
where $(\overline{x},\overline{p})\in T^*\mathbb R^n$, and the effective potential is $U_{\mu}(\overline{x})=\overline U(\overline x)+\frac{\epsilon^2\mu^2}{2}\overline{V'\cdot V'}$. The reduced symplectic structure is given by
\begin{equation}\label{eq:aveParSym}
\omega_{\mu}=d\overline x\wedge d\overline p-\epsilon^3\mu\; d\{\overline{S''V'}^{\flat}\}
\end{equation}
with the above notations for $\overline U, V, S$ and the prime as above, and $\flat$ stands for the lifting indices operator ${\rm Vect}(\mathbb R^n)\rightarrow\Omega^1(\mathbb R^n)$ corresponding to the Euclidean metric on $\mathbb R^n$.
\end{theorem}

\begin{proof}
For the fast time $\tau=t/\epsilon$ let $(\widetilde{x}(\overline{x},\tau,\widetilde{p}(\overline{x},\tau)))$ be a solution of the Hamiltonian system corresponding to the function
$$
 \widetilde{H}(\widetilde{x},\widetilde{p})=\frac 12 \,\widetilde{p}\cdot\widetilde{p}+\widetilde U(\overline x+\widetilde x,\tau)
$$
and satisfying the periodicity and zero mean requirement in $\tau$. (Recall that one splits any solution 
$x=\overline{x}+\widetilde{x}$  into the mean and periodic parts.)
Upon discarding higher order terms in $\epsilon$, the Newton's equation on $\widetilde{x}$ becomes
$$
 d^2\widetilde{x}/dt^2=-\widetilde U'(\overline x+\widetilde x,\tau),
$$
where the prime stands for the derivative with respect to the space variable.
We consecutively obtain
$$
d\widetilde{x}/dt=-\epsilon\int^\tau U'(\overline x+\widetilde x,\tau)\; d\tau=-\epsilon V'
\qquad{\rm and }\qquad
\widetilde{x}=-\epsilon\int^\tau \epsilon\, V'(\overline x+\widetilde x,\tau)\; d\tau=-\epsilon^2 S'\,,
$$
where one integrates by using the zero mean condition on the constants of integration.

Then the averaged metric is given by
$$
((v,\gamma),(v,\gamma))_{(\overline{x},\phi)}=v\cdot v
+2\,\epsilon\gamma\,\langle \overline{S''V'},v\rangle\,(\overline{V'\cdot V'})^{-1}
+ {\gamma^2}(\epsilon^2\overline{V'\cdot V'})^{-1},
$$
where $(v,\gamma)\in T_{\overline{x}} \mathbb R^n\times T_{\phi}\mathbb T\simeq  T_{\overline{x}}\mathbb R^n\times \mathbb R$. Note that the above metric has the form \eqref{eq:aveMet}
discussed in Remark \ref{rem:metric}.

\smallskip

The corresponding fiber inertia operator 
$\mathbb{I}(\overline{x}):\mathfrak t=\mathbb{R}\rightarrow\mathfrak t^*=\mathbb{R},\;\overline{x}\in M$  
is given by $\mathbb{I}(\overline{x})\gamma={\gamma}(\epsilon^2\overline{V'\cdot V'})^{-1}$ for  $\gamma\in\mathbb{R}$.
Then, as follows from Proposition \ref{prop:momentum},
the momentum map $J:T(\mathbb R^n\times \mathbb T)\rightarrow\mathfrak t^*=\mathbb{R}$ 
for $Q=\mathbb R^n\times \mathbb T$ is given by
$J(\bar x,\phi,v,\gamma)=\gamma(\epsilon^2\overline{V'\cdot V'})^{-1}+\epsilon \langle \overline{S''V'},v\rangle\,(\overline{V'\cdot V'})^{-1},$   where $v$ is the image of $p$ under the metric identification.

\smallskip

Finally, the averaged connection $\bar{\mathcal A}\in\Omega^1(\mathbb R^n\times \mathbb T,\mathbb R)$ on the principal $\mathbb T$-bundle $\mathbb R^n\times \mathbb T$ is given by
$\bar{\mathcal A}(\bar x,v,\phi,\gamma)=d\phi+\epsilon^3 \overline{S''V'}^{\flat}.$

We choose $\mu$ to be the value of the momentum map. By Theorem \ref{thm:ave-red}, the reduced symplectic structure on the reduced manifold
$T^*\mathbb R^n=J^{-1}(\mu)/\mathbb T$ has a magnetic term
$$
\omega_{\mu}=d\overline x\wedge d\overline p-\epsilon^3\mu\, d\{\overline{S''V'}d\overline x\}.
$$
The reduced Hamiltonian function on $(T^*\mathbb R^n,\omega_\mu)$ turns out to be
$$
H_{slow}(\overline{x},\overline{p})=\frac 12\, \overline{p}\cdot\overline{p}+U_{\mu}(\overline{x})
$$
at any  $(\overline{x},\overline{p})\in T^*M$,  where the effective potential is
$U_{\mu}(\overline{x})=\overline U(\overline x)+\frac{\epsilon^2\mu^2}{2}\overline{V'\cdot V'}$\,.
\end{proof}

Thus the system obtained by applying the symplectic averaging-reduction procedure is equivalent to the one obtained via the classical averaging method in \cite{cole}.

\begin{remark}
It remains an open question to describe the magnetic term related to the curvature of an appropriate 
$\mathbb T$-bundle in purely geometric terms, similar to the gyroscope description in \cite{cole2}.
\end{remark}

\medskip


\begin{thebibliography}{aa}




\bibitem{ArnMM}  V. I. Arnold. Mathematical Methods of Classical Mechanics. Springer-Verlag, 1989.


\bibitem{arko}  V. I. Arnold, V. V. Kozlov and A. I. Neishtadt. Mathematical Aspects of Classical and Celestial Mechanics. Springer, 2006.


\bibitem{cole}  G. Cox, M. Levi. Magnetic terms in averaged Hamiltonian systems. arXiv:1707.04970.

\bibitem{cole2} G. Cox, M. Levi. Gaussian curvature and gyroscopes. Comm. Pure and Appl. Math., vol. 71:5 (2018), 938-952.

\bibitem{crle}  D. D. Craik and S. Leibovich. A rational model for Langmuir circulations. J. Fluid Mech., vol. 73 (1976), 401-426.

\bibitem{zeit}  P. J. Dellar. Variations on a beta-plane: derivation of non-traditional beta-plane
equations from Hamilton's principle on a sphere. J. of Fluid Mechanics, vol. 674 (2011), 174-195.

\bibitem{kapi}  P. L. Kapitza. Dynamic stability of a pendulum when its point of suspension vibrates. Soviet Phys. JETP, vol. 21 (1951), 588-592.

\bibitem{khch}  B. A. Khesin and Yu. V. Chekanov.  Invariants of the Euler equations for ideal or barotropic hydrodynamics and superconductivity in D dimensions. Physica D, vol. 40 (1989), 119-131.

\bibitem{Klein} F. Klein. Einleitung in die h\"{o}here Geometrie. 3-d Aufl., Berlin, 1926.

\bibitem{land}  L. D. Landau, E.M. Lifshitz. Mechanics. Pergamon, 1979.



\bibitem{mars2} J. E. Marsden, G. Misiolek, J. Ortega, M. Perlmutter, T. S. Ratiu. Hamiltonian Reduction by Stages. Springer-Verlag NY, 2007.


\bibitem{segal} A. Pressley, G. Segal. Loop Groups. Clarendon Press, 1988.


\bibitem{viz}   F. Tiglay, C. Vizman. Euler-Poincar\'{e} equations on Lie groups and homogeneous spaces, their orbit invariants and applications. Lett. Math. Phys., vol. 97:1 (2011), 45-60.

\bibitem{vla}   V. A. Vladimirov, M. R. E. Proctor and D. W. Hughes. Vortex dynamics of oscillating flows. Arnold Math J., vol. 239:2 (2015), 113-126.

\bibitem{yang1}  C. Yang. Multiscale method, central extensions and a generalized Craik-Leibovich equation. J. of Geometry and Physics, vol. 116 (2017), 228-243.

\bibitem{yang2}  C. Yang. On the Hamiltonian and geometric structure of the Craik-Leibovich equation. arxiv:1612.00296.


\end{thebibliography}
\end{document}